\PassOptionsToPackage{sort}{natbib}
\documentclass[preprint,11pt,authoryear]{elsarticle}
\makeatletter
\def\ps@pprintTitle{%
	\let\@oddhead\@empty
	\let\@evenhead\@empty 
	\def\@oddfoot{}%
	\let\@evenfoot\@oddfoot}
\makeatother
\usepackage{multirow}
\usepackage{amsmath,amsfonts,amsthm,amssymb} 
\allowdisplaybreaks
\DeclareMathOperator*{\argmax}{arg\,max}

\usepackage[]{algorithm} 
\usepackage[noend]{algpseudocode}
\usepackage{color,soul}
\usepackage{mathtools,xparse}
\usepackage{enumitem}
\usepackage{easylist}

\usepackage{color}
\usepackage[table,xcdraw]{xcolor}
\definecolor{green}{HTML}{00693E}

\usepackage{booktabs}
\usepackage{bm}
\usepackage{url}
\usepackage{xfrac}
\usepackage{ccaption}
\usepackage[font=small,skip=5pt]{caption}
\usepackage{graphicx}
\usepackage{subcaption}
\usepackage{fancyhdr}
\usepackage[top=1 in, bottom=1 in, left=1 in, right=1 in]{geometry}

\usepackage{tabulary}
\usepackage{mathrsfs}
\usepackage{mathpazo}

\usepackage{setspace} 
\newtheorem{thm}{Theorem}

\newtheorem*{lem}{Lemma}
\newtheorem*{appendixthm}{Theorem}

\graphicspath{{figures/}} 

\begin{document}


\begin{frontmatter}

	\title{Book-Ahead \& Supply Management for Ridesourcing Platforms }
	
	\cortext[cor1]{Corresponding author\newline
	E-mail addresses: cesaryahia@utexas.edu (C.N. Yahia), gustavo@ece.utexas.edu (G. de Veciana), sboyles@mail.utexas.edu (S.D. Boyles), jeanabourahal@utexas.edu (J.A. Rahal), michaelrstecklein@gmail.com (M. Stecklein)  }
	
	\author[label1]{Cesar N. Yahia\corref{cor1}}
	\author[label2]{Gustavo de Veciana}
	\author[label1]{Stephen D. Boyles}
	\author[label2]{Jean Abou Rahal}
	\author[label2]{Michael Stecklein}
	\address[label1]{Department of Civil, Architectural and Environmental Engineering, The University of Texas at Austin}
	\address[label2]{Department of Electrical and Computer Engineering, The University of Texas at Austin \vspace{-25pt}}

	\begin{abstract}
		Ridesourcing platforms recently introduced the ``schedule a ride'' service where passengers may reserve (book-ahead) a ride in advance of their trip. Reservations give platforms precise information that describes the start time and location of anticipated future trips; in turn, platforms can use this information to adjust the availability and spatial distribution of the driver supply. In this article, we propose a framework for modeling/analyzing reservations in time-varying stochastic ridesourcing systems. We consider that the driver supply is distributed over a network of geographic regions and that book-ahead rides have reach time priority over non-reserved rides. First, we propose a state-dependent admission control policy that assigns drivers to passengers; this policy ensures that the reach time service requirement would be attained for book-ahead rides. Second, given the admission control policy and reservations information in each region, we predict the ``target" number of drivers that is required (in the future) to probabilistically guarantee the reach time service requirement for stochastic non-reserved rides. Third, we propose a reactive dispatching/rebalancing mechanism that determines the adjustments to the driver supply that are needed to maintain the targets across regions. For a specific reach time quality of service, simulation results using data from Lyft rides in Manhattan exhibit how the number of idle drivers decreases with the fraction of book-ahead rides. We also observe that the non-stationary demand (ride request) rate varies significantly across time; this rapid variation further illustrates that time-dependent models are needed for operational analysis of ridesourcing systems.
		
	\end{abstract}
	
	
\end{frontmatter}

\section{Introduction}
Recent growth of ridesourcing services is further exacerbating fleet management challenges associated with dynamic and spatially asymmetric passenger demands. Ridesourcing platforms (e.g., Uber and Lyft) need to locate a sufficient number of drivers near anticipated passenger demand to reduce the reach time (i.e., the customer wait time between ride request and the arrival of a driver). However, an abundance of drivers may lead to increased driver idle time. Thus, with the objective of guaranteeing low customer waiting times and low driver idle time, the following questions arise: how many drivers should a ridesourcing platform supply?, and, how should the platform spatially manage idle drivers based on anticipated demand? 

In this article, the primary objective is to investigate the role of book-ahead/reserved rides in the management of driver supply. Reservations give precise information characterizing the start time and location of anticipated trips; in turn, the platform can use this information to adjust the availability and spatial distribution of its driver supply. Thus, given a reach time service requirement that the platform seeks to maintain, we analyze the impact of reservations on the number of drivers supplied throughout the network. Moreover, since passengers that schedule a ride in advance expect the driver to arrive within a desired pickup window, our analysis incorporates such priority of book-ahead rides over non-reserved rides.

In practice, ridesourcing platforms have several control levers that they can use to manage driver supply. These levers include earning guarantees for new drivers, bonuses, and heat maps that show high demand locations where drivers earn more due to surge pricing \citep{Lyft2019a, Lyft2019b}. In addition, as implemented by Lyft in New York City, platforms can restrict the number of active drivers or force them to drive towards high demand areas if they wish to remain online \citep{Lyft2019}.

The proposed supply management framework parallels existing research on ridesourcing systems \citep{wang2019b, lei2019, djavadian2017}. The majority of existing studies assume a fixed number of driver supply and/or steady-state (equilibrium) conditions. However, it is increasingly apparent that demand and supply patterns in ridesourcing systems are time-varying. In addition, these variations in demand and supply occur at a fast pace, and the system may never attain a steady state equilibrium. 

Thus, our proposed framework for analyzing reservations in ridesourcing systems focuses on the \textit{transient} nature of time-varying stochastic demand/supply patterns. Precisely, for any future point in time, we seek to probabilistically characterize the total number of active (non-idle) drivers; this time-dependent probabilistic characterization is determined by the fraction of book-ahead rides, the stochasticity of non-reserved rides, the anticipated time-varying profile of book-ahead rides, and control policies that aim to maintain reach time priority for book-ahead rides. In more detail, as shown in Figure \ref{fig:frame}, the proposed framework consists of the following three components for managing driver supply:

\begin{enumerate}
	\item We develop a state-dependent admission control policy that assigns drivers to passengers. The objective of this control policy is to guarantee the reach time service requirement for book-ahead rides. Effectively, the admission control policy ensures that there is a sufficient number of drivers near the location of anticipated book-ahead rides such that the driver can reach the passenger within the pickup window.
	\item Given this admission control policy and reservations information, we predict the ``target" number of drivers that is required (in the future) to \textit{probabilistically} guarantee the reach time service requirement for stochastic non-reserved rides. The target computations are derived from an upper bound on the time-dependent probability that a non-reserved ride will experience waiting times in excess of the reach time service requirement, and this upper bound can be evaluated using transient analysis of $\text{M}_{t}/\text{GI}/\infty$ queues.
	\item We develop a minimum cost flow driver dispatching/rebalancing mechanism that seeks to maintain the targets across regions. In particular, due to the transition of drivers across geographic regions and the associated passenger demand patterns, the driver supply in a specific region may deviate from the predicted target. Thus, the proposed minimum cost flow mechanism determines the adjustments to the driver supply that are needed to maintain the targets throughout the network.
\end{enumerate}

\begin{figure}
	\centerline{\includegraphics[width=1.1\textwidth]{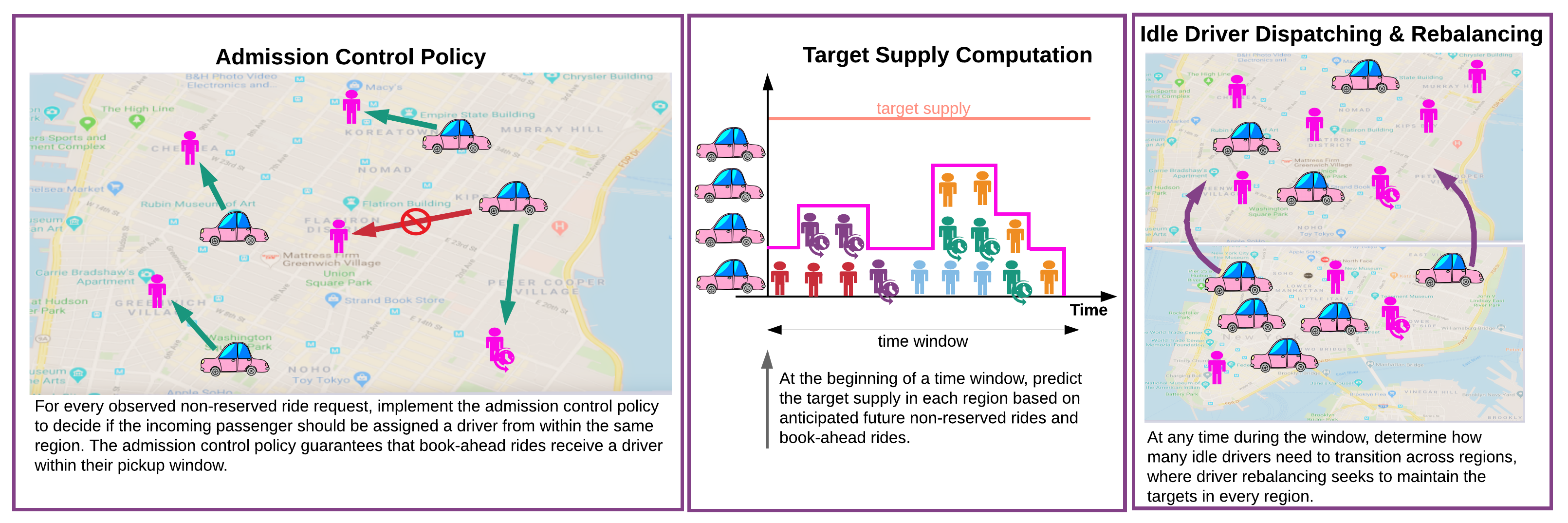}}
	\caption{Proposed framework for assigning drivers to passengers to guarantee the arrival of drivers to book-ahead rides within the pickup window, computing the target supply, and rebalancing drivers across regions to maintain the targets. }
	\label{fig:frame}
\end{figure}

The remainder of this article proceeds as follows: In Section \ref{sec:litrev} we review related work addressing the operation of ridesourcing systems. Section \ref{sec:sysmod} describes the proposed model for analyzing time-dependent ridesourcing dynamics. Section \ref{sec:adcont} presents the admission control policy. Section \ref{sec:targets} derives an upper bound on the performance of the admission control policy and computes the target supply. Section \ref{sec:vehdip} presents the driver dispatching/rebalancing mechanism. Section \ref{sec:manhsim} exhibits simulation results using data from Lyft operations in Manhattan. Section \ref{sec:conc} concludes the article.

\section{Related Work}
\label{sec:litrev}
Ridesourcing platforms are aggressively implementing supply and demand management strategies that drive their expansion into new markets \citep{nie2017}. These strategies can be broadly classified into one or more of the following categories: pricing, fleet sizing, empty vehicle routing (rebalancing), or matching passengers to drivers. Apart from increasing their market share, platforms seek to improve their operational efficiency by minimizing the spatio-temporal mismatch between supply and demand \citep{zuniga-garcia2020}. In this section, we provide a brief survey of existing methods that are used to analyze the operations of ridesourcing platforms.

\subsection{Equilibrium analysis of ridesourcing systems}
The majority of existing studies on ridesourcing systems focus on analyzing interactions between driver supply and passenger demand under \textit{static} equilibrium conditions. These studies seek to evaluate the market share of ridesourcing platforms, competition among platforms, and the impact of ridesourcing platforms on traffic congestion \citep{di2019,bahat2016,wang2018,ban2019,qian2017a}. In addition, following \cite{yang2011}, researchers examined the relationship between customer wait time, driver search time, and the corresponding matching rate at market equilibrium \citep{zha2016, xu2019}. Recently, \cite{di2018} incorporated ridesharing user equilibrium in a network design problem; \cite{zha2018a} proposed an equilibrium model to investigate the impact of surge pricing on driver work hours; \cite{zhang2019a} studied passenger pooling under market equilibrium for different platform objectives and regulations; and \cite{rasulkhani2019} generalized a static many-to-one assignment game that finds equilibrium through matching passengers to a set of routes. While static equilibrium analysis provides valuable strategic decision-making insights, it fails to address stochasticity and time-dependence in ridesourcing dynamics.

\subsection{Steady state analysis of stochasticity in ridesourcing systems}
To investigate stochasticity in demand/supply management, researchers have developed queueing theoretic models for ridesourcing systems. In particular, closed queueing networks were used to analyze rebalancing and pricing policies \citep{banerjee2017,braverman2019, zhang2016}. In these closed queueing networks, the difficulty in designing supply management strategies arises from equilibrium (steady-state) constraints that result in high dimensional non-convex problems \citep{banerjee2017}. Other queueing based approaches include a double-ended queue to characterize stochasticity in matching \citep{xu2019} and an $\text{M}/\text{G}/\text{N}$ queue where each driver is considered to be a server \citep{li2019}. Spatial stochasticity associated with matching was also investigated using Poisson processes to describe the distribution of drivers near a passenger \citep{zhang2019a, zhang2019b, chen2019}. The previously mentioned studies focus on steady-state (equilibrium) analysis that disregards the time-dependent variability in demand/supply patterns. Furthermore, temporal variations in demand/supply patterns may occur rapidly, and the system may not attain the steady-state equilibrium conditions \citep{ozkan2019,braverman2019}. In addition, policies generated from steady-state optimization in closed queueing networks are open-loop (static); this implies that the policies do not react to the time-dependent stochastic state of the system.

\subsection{Time-varying ridesourcing dynamics}
The importance of time dynamics has been emphasized in recent articles that design time-dependent demand/supply management strategies \citep{ramezani2018}. \cite{wang2019} proposed a dynamic user equilibrium approach for determining the optimal time-varying driver compensation rate. Similarly, \cite{nourinejad2019} developed a dynamic model to study pricing strategies; their model allows for pricing strategies that incur losses to the platform over short time periods (driver wage greater than trip fare), and they emphasized that time-invariant static equilibrium models are not capable of analyzing such policies. An alternative dynamic model was proposed by \cite{daganzo2019}; however, the authors focus on the steady-state performance of their model. While these models can be used to analyze time-dependent policies, the authors do not explicitly consider the spatio-temporal stochasticity that results in the mismatch between supply and demand. 

\subsection{Analysis of stochasticity in time-varying ridesourcing dynamics}
The most common approach for analyzing time-dependent stochasticity in ridesourcing systems is to apply steady-state probabilistic analysis over fixed time intervals. However, in the context of driver rebalancing, experimental analysis by \cite{braverman2019} suggests that the time needed to converge to steady-state (equilibrium) in ridesourcing systems is on the order of 10 hours. Thus, since parameters (e.g., passenger arrival rate) vary over much shorter time intervals, the system would not reach the steady-state condition. Subsequently, \cite{braverman2019} proposed a time-dependent look-ahead policy that can be used to make rebalancing decisions at any point in time. Recent studies that addressed operational challenges in ridesourcing systems also advocate for transient analysis instead of steady-state models \citep{ozkan2019, nourinejad2019}.

Another limitation of steady-state policies is that they are independent of the system state. In particular, those policies are based on probabilistic predictions over entire time intervals, and they do not react to the stochastic system state that is realized at a specific time within the time interval. In contrast, state-dependent policies react to the observed fluctuations in the stochastic system state \citep{banerjee2018}.  

Our study falls into this category of analyzing time-dependent stochasticity in ridesourcing systems.
\begin{itemize}
	\item First, we propose a state-dependent admission control policy that reacts to the observed ride requests and available driver supply. This admission control policy ensures that the reach time service requirement is attained for book-ahead rides by choosing which driver to assign to every realized non-reserved ride request.
	\item Second, in a predictive approach over an upcoming time-interval, we provide an upper bound on the performance of the state-dependent admission control policy; precisely, the performance of the policy is measured in terms of the probability that the reach time service requirement would be violated for a non-reserved ride. In contrast to steady-state methods, we use \textit{transient} analysis of $\text{M}_{t}/\text{GI}/\infty$ to determine the aforementioned upper bound at any point in time throughout the window. In other words, we derive a time-dependent upper bound on the probability of reach time violation for non-reserved rides. Subsequently, we use the time-averaged value of the upper bound to compute the ``target" number of drivers that is required during the upcoming time window; thus, this target limits the probability of reach time service violation to be within a desired performance level.
	\item Third, we propose another reactive state-dependent policy for dispatching/rebalancing drivers across multiple regions. Given the predicted ``target" supply for an upcoming time window, the minimum cost flow dispatching/rebalancing policy seeks to maintain the targets across multiple regions. For a specific system state at some time within the time window, the dispatching/rebalancing mechanism determines the number of idle drivers that should transition to adjacent regions to maintain the targets.
\end{itemize}

\begingroup

\begin{table}
	\fontsize{9.5pt}{12pt}\selectfont
	\centering
	\caption{Table of Notation \& Definitions}
	\begin{tabular}{r c p{12cm}}
		\hline
		active driver & $\triangleq$ & drivers are active from the moment they are dispatched to pick up a passenger and until the passenger leaves the vehicle \\
		idle driver & $\triangleq$ & driver waiting to be dispatched (not active)\\
		ride initiation/start & $\triangleq$ & time driver is dispatched to pick up passenger\\
		ride completion & $\triangleq$ & time passenger leaves vehicle\\
		ride duration & $\triangleq$ & total time while driver is active (includes pick up time)\\
		$R$  & $\triangleq$ & set of regions $\{1,.. ,r,.., m\}$\\
		window $k$ & $\triangleq$ & time window $ \left( kw, (k+1)w\right]$ \\
		$w$ &  $\triangleq$ &  duration of time window\\
		$c^{k}_{r}$ & $\triangleq$ & target number of drivers in region $r$ during window $k$ that would probabilistically guarantee a desired reach time service level\\	
		$f_{r}^{P,k}(t)$ & $\triangleq$ & deterministic process representing active drivers at time $t\in \left( kw, (k+1)w\right]$ that are serving requests which initiated in $r$ during \textit{previous} time windows\\
		$f_{r}^{BA,k}(t)$ & $\triangleq$ & deterministic process representing active drivers at time $t\in \left( kw, (k+1)w\right]$ that are associated with \textit{book-ahead} trips that initiate within window $\left( kw, (k+1)w\right]$ in region $r$\\
		$N^{k}_{r}(t)$  & $\triangleq$ & stochastic process representing active drivers at time $t\in \left( kw, (k+1)w\right]$ that are associated with \textit{admitted} stochastic non-reserved rides that initiate within window $\left( kw, (k+1)w\right]$ in region $r$\\
		$\lambda^{k}_{r}(t)$ & $\triangleq$ & demand rate at which stochastic non-reserved ride requests initiate during window $k$ in region $r$\\
		$g^{k}_{r}(\cdot)$ & $\triangleq$ & probability density function characterizing the ride duration (completion time - trip request time) of stochastic non-reserved rides that appear during window $k$ in region $r$\\
		$G^{k}_{r}(\cdot)$ & $\triangleq$ & cumulative density function of $g^{k}_{r}(\cdot)$\\
		$f_{r}^{A(\tau_{i}),k}(t)$ & $\triangleq$ & active drivers at time $t\in \left(\tau_{i}, \min\{\tau_{i}+D_{i}, (k+1)w \}  \right]$ corresponding to non-reserved rides that were \textit{previously admitted} between $\left( kw, \tau_{i} \right]$ in region $r$\\
		$\tau_{i}$ & $\triangleq$ & arrival time of the $i^{\text{th}}$ non-reserved ride request\\
		$D_{i}$ & $\triangleq$ & ride duration of the $i^{\text{th}}$ non-reserved ride\\
		$\gamma_{i}$ & $\triangleq$ & indicator function/random variable characterizing the event that the $i^{\text{th}}$ non-reserved ride request is admitted\\
		$B_{r}^{k}$ & $\triangleq$ &  average blocking probability during window $k$ in region $r$\\
		$\delta$ & $\triangleq$ & desired reach time quality of service for non-reserved rides (upper bound on the average blocking probability)\\
		$N^{k,\infty}_{r}  \left(  t' \right)$ & $\triangleq$ & number of busy servers at time $t'\in \left( 0, w\right]$ in a transient $\text{M}_{t}/\text{GI}/\infty$ queue that starts empty at $t'=0$; equivalently, the number of active non-reserved rides assuming that all stochastic non-reserved requests are admitted \\
		$\rho^{k}_{r}(t')$ & $\triangleq$ & time-dependent mean/variance of the Poisson distribution characterizing $N^{k,\infty}_{r}  \left(  t' \right)$ at time $t'\in \left( 0, w\right]$\\
		$a_{r}$ & $\triangleq$ & number of active drivers in region $r$\\
		$e_{r}$ & $\triangleq$ & number of idle drivers in region $r$\\
		$s_{r}^{v}$ & $\triangleq$ & virtual supply in region $r$ representing drivers in excess of the target $c^{k}_{r}$ that can be removed from region $r$\\
		$d_{r}^{v}$ & $\triangleq$ & virtual demand in region $r$ representing drivers that should be added to region $r$ to meet the target $c^{k}_{r}$\\
		$\Delta_{r}$ & $\triangleq$ & if region $r$ has virtual demand, then $\Delta_{r}=-d_{r}^{v}$; otherwise, if the region has virtual supply, then $\Delta_{r}=s_{r}^{v}$ \\
		$h_{ij}$ & $\triangleq$ & recommended driver transitions between region $i$ and $j$\\
		$\mathbf{1}\{  \cdot\}$ & $\triangleq$ & indicator function or random variable\\
		\hline
	\end{tabular}
	\label{tab:TableOfNotations}
\end{table}

\endgroup

\section{System Model}
\label{sec:sysmod} 
In this section, we describe a general model for time-varying dynamics in ridesourcing systems. The proposed model represents the number of future \textit{active} rides that initiate in a region. A ride/driver is active from the moment the driver is dispatched to pick up the passenger until the trip is completed. For non-reserved rides, the ride becomes active at the same time as the request is initiated. On the other hand, for book-ahead rides, there is a lag between the time that the request is initiated and the time that the drivers is dispatched to pick up the passenger. While active, drivers are associated with the passenger and can not take on other requests. The ride duration (service time) is the time spent while the driver is active which includes the pick up time. A ride starts when the driver becomes active and ends when the driver is idle again.

The active rides are represented over a set of geographic regions $R=\{1,.. , m\}$. These regions are sufficiently small that if a ride request initiates in a region and the assigned driver is operating in the same region, then the reach time is within a desired service level. In other words, if we want the reach time to be under 10 minutes, then the time it takes to drive from any point to any other point within the defined region should be under 10 minutes. 

Consequently, we incorporate reservations by providing reach-time priority for book-ahead rides. In particular, for a driver to arrive within the book-ahead ride pickup window, the driver must be geographically close to the passenger at the anticipated trip start time. Thus, we consider that book-ahead ride requests must be assigned a driver from within the same region in which the request initiates, and that satisfying the reach time service requirement for book-ahead rides is equivalent to a driver arriving to the passenger within the pickup window. In Section \ref{sec:adcont}, we design an admission control policy that guarantees that book-ahead rides will be assigned a driver from within the same region.   

In the proposed ridesourcing model, we do not explicitly analyze ridesharing (i.e., passenger pooling); however, the predicted number of active rides would be a conservative estimate on the corresponding value in ridesharing systems. Furthermore, for tractable target computations, we examine each region separately. In other words, the admission control and corresponding targets assume passengers remain within the zone, disregarding the variation in destinations. Then, to account for the spatial distribution of passenger destinations and the associated movement of drivers across regions, we implement a min-cost flow rebalancing methods that maintains the targets across regions. Note that the targets themselves represent a desired number of drivers that is determined by passenger demand; this implies that the targets do not depend on the stochasticity of drivers entering and exiting the system.

We proceed by describing the model for active rides in each region. For each region, this model consists of processes representing book-ahead rides and non-reserved stochastic rides. The processes form the basis of subsequent sections that discuss the admission control policy and the computation of targets.

\subsection{Time-varying profiles representing rides that will be active in the future}
\label{sec:sysmod-proc}

\begin{figure}
	\centerline{\includegraphics[width=0.8\textwidth]{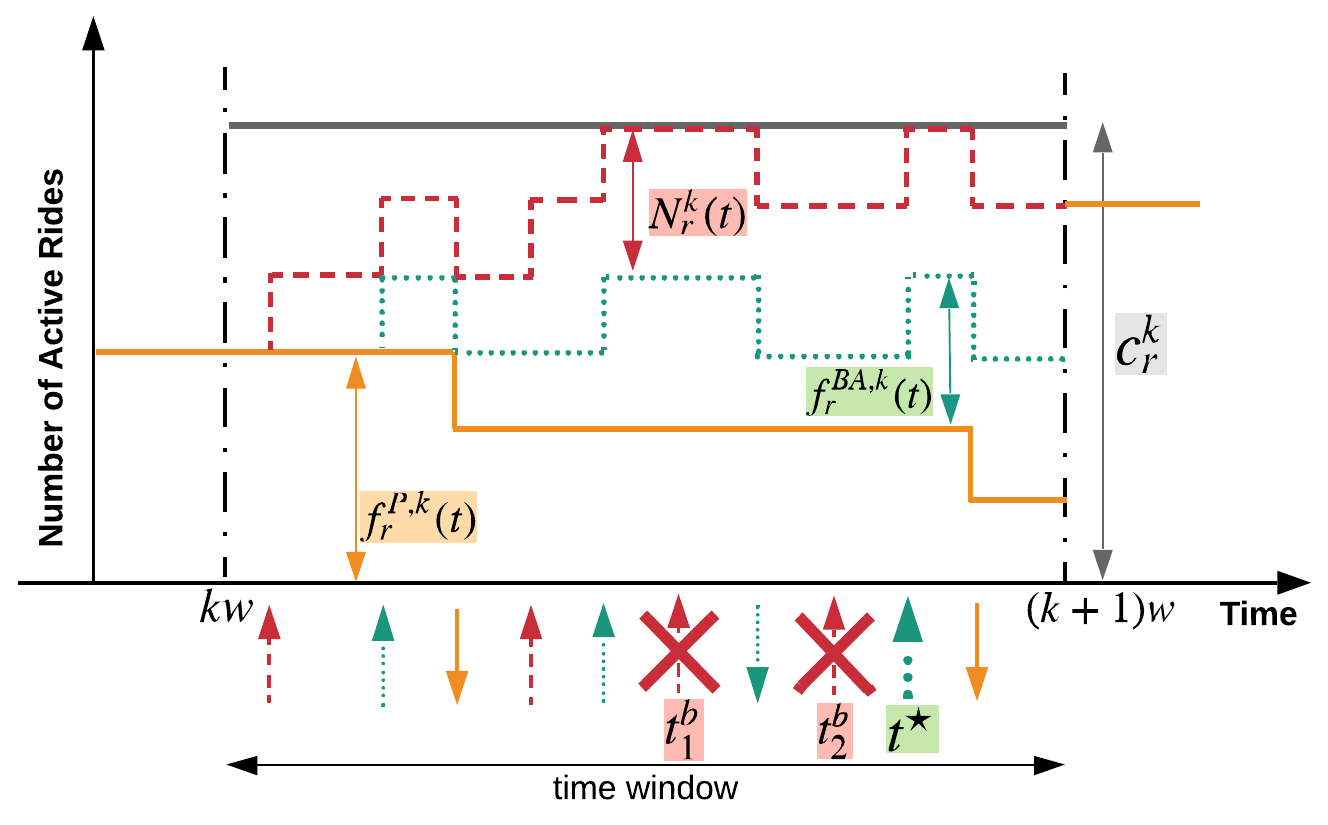}}
	\caption{System model characterizing the \textit{cumulative} number of rides that will be active in the future at time $t\in\left( kw, (k+1)w\right]$. Arrows pointing upwards indicate ride start time. Arrows pointing downwards indicate ride completion. Solid lines correspond to $f_{r}^{P,k}(t)$, dotted lines correspond to $f_{r}^{BA,k}(t)$, and dashed lines correspond to $N^{k}_{r}(t)$. Non-reserved requests marked with an ``X" are blocked requests.}
	\label{fig:profile}
\end{figure}

In each region $r \in R$, we represent ridesourcing dynamics over future time windows of length $w$. At the beginning of each window $k$, corresponding to time interval $ \left( kw, (k+1)w\right]$, the ridesourcing platform 
can characterize three processes (two deterministic and one stochastic) that will be realized during the upcoming window $ \left( kw, (k+1)w\right]$. The processes represent \textit{active} drivers at time $t \in \left( kw, (k+1)w\right]$ that are serving requests initiated within the region.

First, we assume that the platform knows the anticipated start time for \textit{book-ahead} rides that will initiate during window $k$. We also assume that the platform can accurately estimate the corresponding ride duration (i.e. the platform has full trip information for future book-ahead rides). Thus, at the start of window $k$, the platform can characterize the \textit{deterministic} process $\{f_{r}^{BA,k}(t):t\in \left( kw, (k+1)w\right] \}$ that represents the number of active drivers at time $t$ associated with book-ahead trips that will initiate in region $r$ within window $k$.

Second, at the beginning of time window $\left( kw, (k+1)w\right]$, currently active drivers serving rides that started in region $r$ prior to time $t=kw$ are known to the platform. For those \textit{previously observed} trips, we assume that the platform can accurately estimate the trip completion time. Thus, at the start of window $k$, the platform can characterize the deterministic process $\{f_{r}^{P,k}(t):t\in \left( kw, (k+1)w\right] \}$. This process represents the number of active drivers at time $t$ that are serving rides started in region $r$ during previous time windows. In other words, those are previously observed rides that haven't ended yet and may correspond to either passenger type (book-ahead or non-reserved).

Third, at the beginning of window $k$, the platform also anticipates \textit{non-reserved} stochastic rides that will arise throughout the upcoming window in region $r$. For those rides, we assume that the platform can estimate the demand (ride request) rate $\{\lambda^{k}_{r}(t):t\in \left( kw, (k+1)w\right] \} $. We also assume that the platform can estimate a general distribution $g^{k}_{r}(\cdot)$ that corresponds to the ride duration (the CDF of $g^{k}_{r}(\cdot)$ is $G^{k}_{r}(\cdot)$), and we consider that the duration of any specific non-reserved trip is independent of other trips. Then, we define a stochastic process $\{N^{k}_{r}(t):t\in\left( kw, (k+1)w\right] \}$ that represents the number of active drivers at time $t$ associated with \textit{admitted} stochastic rides which initiate in region $r$ during window $k$. In this case, a non-reserved ride request would be admitted if it is assigned a driver from within the same region.

The deterministic processes  $\{f_{r}^{P,k}(t), f_{r}^{BA,k}(t) : t\in \left( kw, (k+1)w\right] \}$ and the stochastic process $ \{  N^{k}_{r}(t) : t\in \left( kw, (k+1)w\right] \} $ are illustrated in Figure \ref{fig:profile}. The figure shows the \textit{cumulative} number of active drivers at time $t\in \left( kw, (k+1)w\right] \}$.

\begin{figure}
	\centerline{\includegraphics[width=0.8\textwidth]{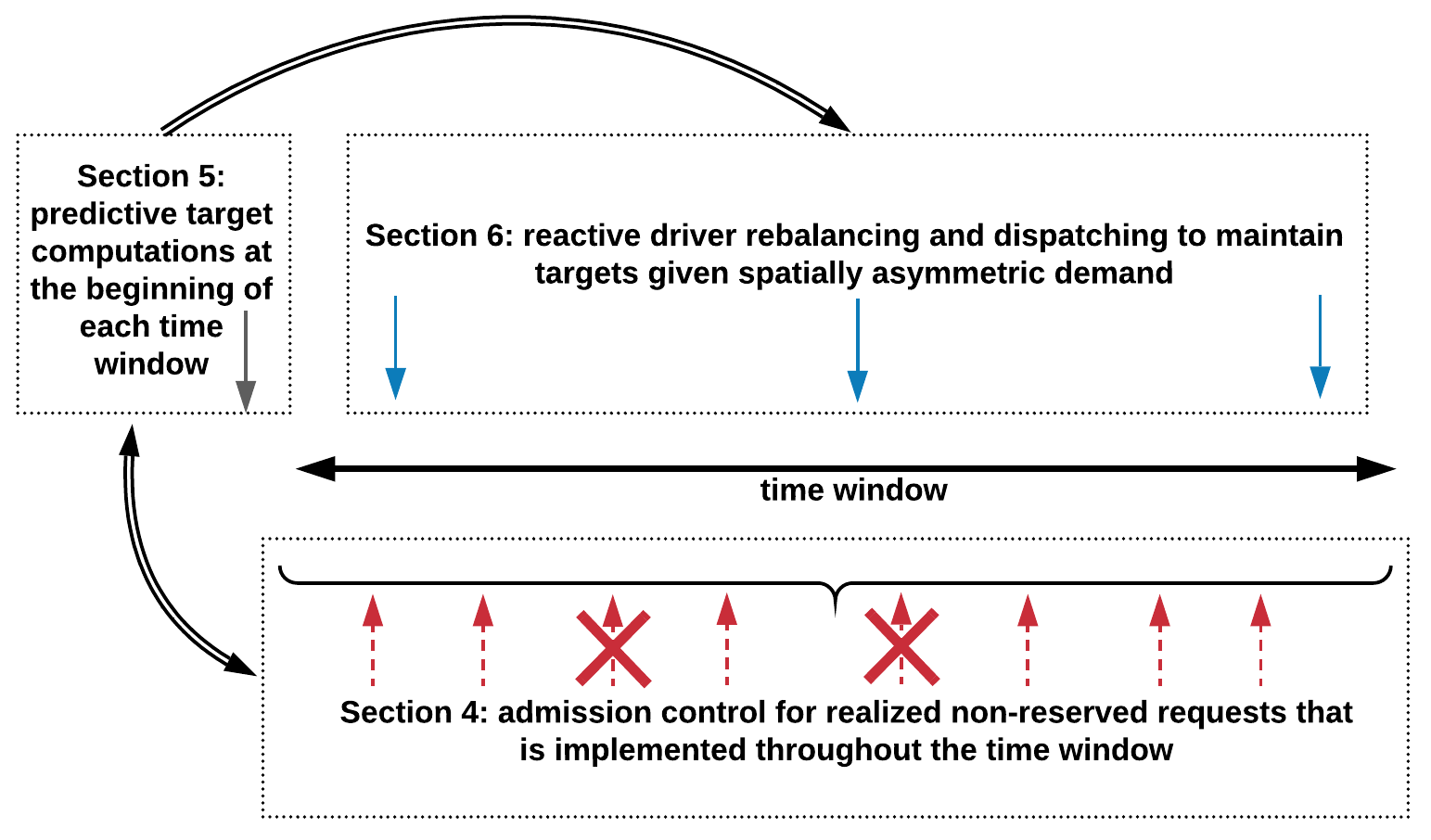}}
	\caption{ Implementation of the proposed framework across time windows.  }
	\label{fig:rel}
\end{figure}

The next section describes the admission control policy that decides whether to admit non-reserved rides based on the difference between the predicted targets and the number of active drivers. The admission control policy is state-dependent such that the admission decision is determined for each ride request once the request is observed. In more detail, the admission decision depends on the current \textit{known} state of the system for the entire duration that the observed ride will be active. Given this policy, we discuss in Section \ref{sec:targets} how the targets are evaluated at the beginning of the window. However, to compute the targets, we refer to the \textit{predicted} future system state under the control policy, and we resort to a probabilistic characterization of the anticipated non-reserved rides (i.e., we further analyze the stochastic process $\{N^{k}_{r}(t):t\in\left( kw, (k+1)w\right] \}$). In other words, the admission control policy uses the targets in determining the \textit{deterministic} admission decisions while the targets are evaluated using the predicted \textit{stochastic} system state that will arise under the control policy. Then, in Section \ref{sec:vehdip}, we present the driver dispatching and rebalancing mechanism that maintains the targets given the \textit{observed} demand patterns. Figure \ref{fig:rel} illustrates the relationship between different components of this article and the time at which those components would be implemented.

\section{Admission Control Policy}
\label{sec:adcont}
In this section, we present an admission control policy that is used to assign drivers to realized non-reserved ride requests. In each region, when a non-reserved ride request is observed, the proposed state-dependent control policy determines whether the request should be \textit{admitted} or \textit{blocked}. If the request is admitted, then a driver from within the same region is assigned to serve the request. 

The admission decision is based on the supply in the region, the anticipated book-ahead rides, and the previously admitted non-reserved rides. In particular, the policy seeks to guarantee that a driver from within the same region would be available to serve anticipated future book-ahead rides. Thus, admission control aims to guarantee that drivers arrive within the pickup window for future book-ahead rides. Since the same policy is implemented for each region, we restrict our discussion in this section to a single region $r\in R$.

At any time $t\in \left( kw, (k+1)w\right]$, the admission control policy determines if idle drivers will be available in the region by comparing the number of active rides to the \textit{target supply} $c^{k}_{r}$. The target supply $c^{k}_{r}$, illustrated in Figure \ref{fig:profile}, is the total number of drivers associated with region $r$ during window $k$; this total includes drivers that are serving ride requests initiated in region $r$ \textit{and} drivers idling in region $r$. The target $c^{k}_{r}$ represents a desired level of driver supply that would probabilistically guarantee the reach time service requirement for non-reserved rides (Section \ref{sec:targets}). The admission control policy assumes that the targets $c^{k}_{r}$ will be maintained in each region $r$ throughout the time window $k$. For tractable computation, the admission control policy also assumes that the passengers destinations remain within the region (in Section \ref{sec:vehdip}, we devise a driver dispatching/rebalancing mechanism that considers the spatial distribution of demand and seeks to maintain the target across regions).

\subsection{Policy Implementation}
A non-reserved ride request is admitted if, upon admission, the total number of active rides does not exceed the target supply for the entire ride duration. Once a non-reserved ride request is observed, the associated ride duration would be also revealed to the platform. Then, there are two cases where the admission control policy would \textit{block} the non-reserved ride request: (1) There are not enough available drivers within the region at the time of request initiation; this is illustrated in Figure \ref{fig:profile} at time $t^{b}_{1}$, where the sum $N^{k}_{r}(t^{b}_{1})+f_{r}^{BA,k}(t^{b}_{1})+f_{r}^{P,k}(t^{b}_{1})$ is equal to the target $c^{k}_{r}$. In other words, admission of the non-reserved ride would result in the total number of active rides \textit{exceeding} the target supply at the time of request initiation. (2) Admission of the non-reserved ride would result in reach time service violation for an anticipated book-ahead ride; in Figure \ref{fig:profile}, admission of the non-reserved ride request that initiates at time $t^{b}_{2}$ would lead to reach time violation for the book-ahead trip that initiates at $t^{\star}$ (considering that the observed ride duration of the request that initiates at $t^{b}_{2}$ extends beyond $t^{\star}$). In other words, if the non-reserved ride was admitted at $t^{b}_{2}$, then at $t^{\star}$ (just before the book-ahead request is anticipated) the sum $N^{k}_{r}(t^{\star})+f_{r}^{BA,k}(t^{\star})+f_{r}^{P,k}(t^{\star})$ would be equal to the target supply $c^{k}_{r}$; this implies that the total number of active rides would \textit{exceed} the target supply when the book-ahead ride at  $t^{\star}$ starts (equivalently, the book-ahead ride would not be assigned a driver from within the same region). 

In more detail, let $\tau_{i}$ be the arrival time of the $i^{\text{th}}$ non-reserved ride request, and let $D_{i}$ be the corresponding ride duration. In addition, let $\gamma_{i}$ be an indicator function that takes the value one if the $i^{\text{th}}$ non-reserved ride request is admitted. Equation \ref{eqn:gamma} gives the expression for $\gamma_{i}$ (i.e., Equation \ref{eqn:gamma} represents the condition for admission). In Equation \ref{eqn:gamma}, $f^{A(\tau_{i}),k}_{r}(t)$ represents \textit{previously} \textit{admitted} non-reserved rides that would be active at time $ t\in \left( \tau_{i}, \min \left\lbrace \tau_{i}+D_{i}, (k+1)w \right\rbrace \right] $. In other words, $f^{A(\tau_{i}),k}_{r}(t)$ represents previously admitted non-reserved rides that would be active during the time that the $i^{\text{th}}$ non-reserved ride request is being served. Note that the projected ride duration of the $i^{\text{th}}$ non-reserved user is restricted to $t \in \left( \tau_{i}, \min\{\tau_{i}+D_{i}, (k+1)w \} \right] $ instead of $t \in \left( \tau_{i}, \tau_{i}+D_{i} \right] $ since admission control decisions are made per window $k$ (i.e., the rides whose duration extends beyond $t=(k+1)w$ would become part of $f_{r}^{P,k+1}(t)$).
\begin{equation}
\begin{aligned}
\gamma_{i}=\mathbf{1} \left\lbrace 1+ f_{r}^{P,k}(t) + f_{r}^{BA,k}(t) + f^{A(\tau_{i}),k}_{r}(t) \leq c^{k}_{r}, \quad \forall t\in \left( \tau_{i}, \min \left\lbrace \tau_{i}+D_{i}, (k+1)w \right\rbrace \right] \right\rbrace
\end{aligned}
\label{eqn:gamma}
\end{equation}

If we let $\tau_{n}$ and $D_{n}$ be the arrival time and ride duration of the $n^{\text{th}}$ previously observed non-reserved ride (where $n\in \{1,..., i-1\}$), we can express $f^{A(\tau_{i}),k}_{r}(t)$ as shown in Equation \ref{eqn:prevAdm}. In this equation, $\mathbf{1}\{\tau_{n}+D_{n}>t \}$ takes the value one if the $n^{\text{th}}$ previously observed non-reserved ride would be active at time $t$, and $\gamma_{n}$ takes the value one if the $n^{\text{th}}$ non-reserved request was admitted.
\begin{equation}
\begin{aligned}
f^{A(\tau_{i}),k}_{r}(t)= \sum_{n=1}^{i-1} \mathbf{1}\{ \tau_{n}+D_{n}>t \}\gamma_{n}, \quad  t\in \left( \tau_{i}, \min \left\lbrace \tau_{i}+D_{i}, (k+1)w \right\rbrace \right] 
\end{aligned}
\label{eqn:prevAdm}
\end{equation}

We emphasize that the control policy is state-dependent and applied upon the receipt of each ride request; this implies that the state of the system is deterministic and all the variables (including $\tau_{n}, D_{n},\gamma_{n},f^{A(\tau_{i}),k}_{r}(t),\tau_{i}, D_{i},\gamma_{i}$) are known at time $\tau_{i}$. Then, the admission decision for the $i^{\text{th}}$ non-reserved user follows directly from evaluating expressions \ref{eqn:gamma} and \ref{eqn:prevAdm}.


A non-reserved ride request that is blocked may be assigned a driver from an external region (i.e., the passenger will experience a long wait time). Alternatively, blocked non-reserved requests may be dropped from the system, where this indicates a passenger canceling the ride due to the extended wait time. In the simulation experiments (Section \ref{sec:manhsim}), we follow the latter approach.

\section{Target Supply for Probabilistically Guaranteeing the Reach Time Quality of Service}
\label{sec:targets}
While the admission control policy is a state-dependent policy that is applied during the time window $\left( kw, (k+1)w\right]$, it is based on the target supply $c^{k}_{r}$ that is determined at the beginning of the time window $t=kw$. For a specific region $r$, the target $c^{k}_{r}$ represents the total number of drivers that is required during window $k$ to probabilistically guarantee the reach time service requirement for non-reserved rides. Drivers are considered to be associated with a region if they are either serving requests that initiated in the region or they are idle within the region. In this section, we discuss how the targets can be computed at the beginning of the time window. First, we derive a time-dependent upper bound on the blocking probability corresponding to the admission control policy. Then, we determine the target number of drivers that limits the time-averaged blocking probability to be below a certain quality of service threshold. In turn, limiting the time-averaged blocking probability is equivalent to limiting the probability of reach time violation for non-reserved ride requests.

In Equations \ref{eqn:gamma} and \ref{eqn:prevAdm}, representing the admission control policy when the $i^{\text{th}}$ non-reserved ride request is received, the values of all the variables are known (for every non-reserved ride request that was previously received, the trip information would have been revealed to the platform). However, at the beginning of the time window, the platform would not know the arrival time, ride duration, and admission decision of a future non-reserved request. Therefore, at the beginning of the time window, $\tau_{n}, D_{n},\gamma_{n},f^{A(\tau_{i}),k}_{r}(t),\tau_{i}, D_{i},\gamma_{i}$ are all random variables. To express the probability of admission, we can re-write Equation \ref{eqn:gamma} as shown in Equation \ref{eqn:padm}. Hence, Equation \ref{eqn:block} represents the probability that the $i^{\text{th}}$ non-reserved ride request would be blocked.
\begin{equation}
\begin{aligned}
P(\gamma_{i}=1)=P\left( 1+ f_{r}^{P,k}(t) + f_{r}^{BA,k}(t) + f^{A(\tau_{i}),k}_{r}(t) \leq c^{k}_{r}, \quad \forall t\in \left( \tau_{i}, \min \left\lbrace \tau_{i}+D_{i}, (k+1)w \right\rbrace \right] \right)
\end{aligned}
\label{eqn:padm}
\end{equation}
\begin{equation}
\begin{aligned}
&P(\gamma_{i}=0)=1-P(\gamma_{i}=1)=\\ 
&P\left( \exists t\in \left(\tau_{i}, \min\{\tau_{i}+D_{i}, (k+1)w \}  \right]: 1+ f_{r}^{P,k}(t) + f_{r}^{BA,k}(t) + f^{A(\tau_{i}),k}_{r}(t) > c^{k}_{r} \right)= \\ 
&P\left( \exists t\in \left(\tau_{i}, \min\{\tau_{i}+D_{i}, (k+1)w \}  \right]: 1+ f_{r}^{P,k}(t) + f_{r}^{BA,k}(t) + \sum_{n=1}^{i-1} \mathbf{1}\{ \tau_{n}+D_{n}>t \}\gamma_{n} > c^{k}_{r} \right)
\end{aligned}
\label{eqn:block}
\end{equation}

Observe that for \textit{predictive} target computations, $f^{A(\tau_{i}),k}_{r}(t)= \sum_{n=1}^{i-1} \mathbf{1}\{ \tau_{n}+D_{n}>t \}\gamma_{n}$ represents stochastic non-reserved ride requests that will be admitted between $\left( kw, \tau_{i}\right]$ \textit{and} will be active at time $t\in \left( \tau_{i}, \min \left\lbrace \tau_{i}+D_{i}, (k+1)w \right\rbrace \right]$. Recall that future stochastic non-reserved ride requests appear at a demand rate $\{\lambda^{k}_{r}(t):t\in \left( kw, (k+1)w\right] \} $ and the corresponding ride durations are generally distributed according to a distribution $g^{k}_{r}(\cdot)$. Previously, we defined the stochastic process $\{N^{k}_{r}(t):t\in\left( kw, (k+1)w\right] \}$ that represents the number of future active drivers associated with \textit{admitted} non-reserved rides. Notice that $N^{k}_{r}(\tau_{i})=f^{A(\tau_{i}),k}_{r}(\tau_{i})$ is the number of admitted non-reserved ride requests that will be active at time $\tau_{i}$. However, for $t\in \left( \tau_{i}, \min \left\lbrace \tau_{i}+D_{i}, (k+1)w \right\rbrace \right]$, $N^{k}_{r}(t)\neq f^{A(\tau_{i}),k}_{r}(t)$ since $N^{k}_{r}(t)$ includes non-reserved ride requests that will be admitted between $\left( kw,t \right]$ while $f^{A(\tau_{i}),k}_{r}(t)$ is restricted to non-reserved ride requests admitted between $\left( kw, \tau_{i}\right]$.

To determine the target supply $c^{k}_{r}$, we need to evaluate the blocking probability expression in Equation \ref{eqn:block} for different values of $c^{k}_{r}$. However, this probability expression is difficult to analyze due to the dependence of $\gamma_{i}$ (admission of $i^{\text{th}}$ non-reserved request) on the random variables $\tau_{n}, D_{n}$ (arrival time, ride duration) and $\gamma_{n}$ (admission) associated with previously arriving non-reserved ride requests $n \in \{1,..., i-1\}$. In addition, the arrival time $\tau_{i}$ of the $i^{\text{th}}$ non-reserved ride request also depends on the arrival time $\tau_{n}$ of all previous requests. Moreover, the correlations between the random variables have to be considered over the entire time interval $\left(\tau_{i}, \min\{\tau_{i}+D_{i}, (k+1)w \}  \right]$ and this interval also has time-varying functions $f_{r}^{P,k}(t)$ and $f_{r}^{BA,k}(t)$ that impact the admission probability.

Thus, instead of attempting to evaluate Equation \ref{eqn:block}, we provide an upper bound on the blocking probability. In particular, let $\{ N^{k,\infty}_{r}(t): t\in \left( kw, (k+1)w\right] \}$ be the number of busy servers in a \textit{transient} $\text{M}_{t}/\text{GI}/\infty$ queue that starts empty at the beginning of the window $t=kw$, where the arrivals to the $\text{M}_{t}/\text{GI}/\infty$ queue appear according to a Poisson process with rate $\{\lambda^{k}_{r}(t):t\in \left( kw, (k+1)w\right] \} $ and the service distribution is $g^{k}_{r}(\cdot)$.\\

\begin{thm}
	\label{thm:ub}
	The blocking probability, $P(\gamma_{i}=0)$, for the $i^{\text{th}}$ stochastic non-reserved ride request that appears at time $\tau_{i}$ is bounded above by $P\left( N^{k,\infty}_{r}  \left(  \tau_{i} \right) \geq  c^{k}_{r} - \displaystyle \max_{  t\in\left( \tau_{i}, (k+1)w \right] } \left[ f_{r}^{P,k}(t) + f_{r}^{BA,k}(t) \right]  \right)   $
\end{thm}
\begin{proof}
	See Appendix A.
\end{proof}

Given this upper bound in Theorem \ref{thm:ub}, we can limit the blocking probability at time $\tau_{i}$ to be below a certain quality of service threshold $\delta$ by ensuring that the upper bound is below $\delta$ (as shown in Inequality \ref{eqn:upperbound}). Importantly, while $P(\gamma_{i}=0)$ is difficult to evaluate as mentioned earlier, the upper bound can be evaluated for any value $c^{k}_{r}$ and at any time $\tau_{i}$ using transient analysis of $\text{M}_{t}/\text{GI}/\infty$ queues (Section \ref{sec:transanal}). Subsequently, after illustrating how the upper bound can be evaluated at any time for a specific value of $c^{k}_{r}$, we discuss (Section \ref{sec:avgblock}) how to use this upper bound to determine the target supply, where the target supply is the minimal $c^{k}_{r}$ that limits the time-averaged blocking probability to be below the threshold $\delta$.

\begin{equation}
\begin{aligned}
P(\gamma_{i}=0) \leq P\left( N^{k,\infty}_{r}  \left(  \tau_{i} \right) \geq  c^{k}_{r} - \displaystyle \max_{  t\in\left( \tau_{i}, (k+1)w \right] } \left[ f_{r}^{P,k}(t) + f_{r}^{BA,k}(t) \right]  \right) \leq \delta
\end{aligned}
\label{eqn:upperbound}
\end{equation}

\subsection{Time-Dependent Distribution of the Number of Busy Servers in an $\text{M}_{t}/\text{GI}/\infty$ Queue}
\label{sec:transanal}
To evaluate the upper bound $P\left( N^{k,\infty}_{r}  \left(  \tau_{i} \right) \geq  c^{k}_{r} - \displaystyle \max_{  t\in\left( \tau_{i}, (k+1)w \right] } \left[ f_{r}^{P,k}(t) + f_{r}^{BA,k}(t) \right]  \right)$ at time $\tau_{i}$ and for a specific $c^{k}_{r}$, we use a graphical approach that was first recognized by \cite{prekopa1958} and was subsequently further discussed in articles that analyze $\text{M}_{t}/\text{GI}/\infty$ queues \citep{foley1982, eick1993}. We show that the number of busy servers in an $\text{M}_{t}/\text{GI}/\infty$ queue that starts empty, $N^{k,\infty}_{r}  \left(  \tau_{i} \right)$, has a \textit{time-dependent} Poisson distribution, and we derive the time-dependent mean associated with this distribution. Thus, since $\displaystyle \max_{  t\in\left( \tau_{i}, (k+1)w \right] } \left[ f_{r}^{P,k}(t) + f_{r}^{BA,k}(t) \right] $ and $c^{k}_{r}$ are known values at time $\tau_{i}$, evaluating the upper bound is equivalent to computing the probability that a Poisson random variable is greater than or equal to a constant.

Referring to Figure \ref{fig:mginf}, consider stochastic arrivals to an $\text{M}_{t}/\text{GI}/\infty$ queue such that $x_{j}$ denotes the $j^{\text{th}}$ arrival time according to the Poisson process and $s_{j}$ denotes the corresponding generally distributed service time. In time window $\left( kw, (k+1)w\right]$, the $\text{M}_{t}/\text{GI}/\infty$ queue is \textit{initially empty} at time $kw$. 

\begin{figure}
	\centerline{\includegraphics[width=0.6\textwidth]{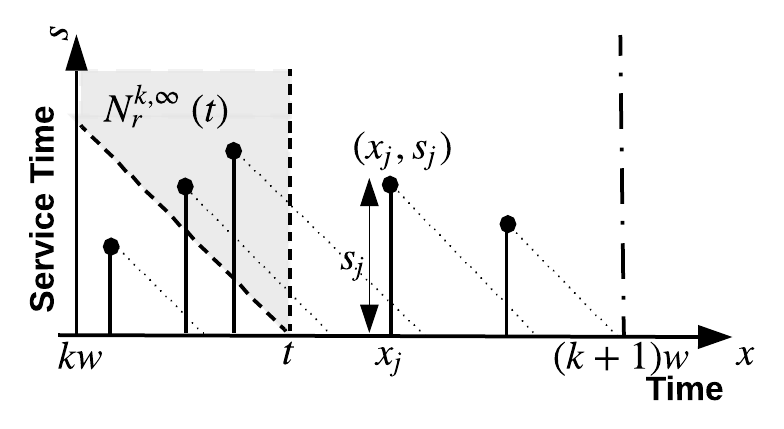}}
	\caption{Service time vs. arrival time associated with a transient $\text{M}_{t}/\text{GI}/\infty$ queue that starts empty at time $kw$. Since there are an infinite number of servers, all arrivals start being serviced immediately. The dotted diagonal lines represent the decrease in remaining service time as the user is being served. For any time $t$, the number of users still being served is equal to the number of diagonal lines that intersect a vertical line from $t$; equivalently, the number of users still being served at $t$ is the number of points in the shaded area.}
	\label{fig:mginf}
\end{figure}

We can think of $(x_{j}, s_{j})$ as a random point in the two-dimensional plane $\left( kw, (k+1)w\right] \times [0,\infty)$ that represents the arrival time and service duration. For any two-dimensional set $S$ in $\left( kw, (k+1)w\right]\times [0,\infty)$, the number of points in the set represents random sampling of the arrivals Poisson process; thus, the number of points in the set $S$ is \textit{Poisson distributed}. We also know that disjoint two-dimensional sets correspond to independent sampling of a Poisson process; this implies that the number of points in each set is independent of other disjoint sets. 

Furthermore, considering an infinitesimal two-dimensional square set with an area $ds(dx)$, we can see that the mean number of points in that set is $\lambda^{k}_{r}(x) dx \left(g^{k}_{r}(s) (ds)\right)$; this implies that the intensity of the two-dimensional Poisson distribution is $\lambda^{k}_{r}(x)g^{k}_{r}(s)$. Thus, the distribution of points defined as (arrival time, service duration) is Poisson over the two-dimensional space, and the \textit{mean} number of points for any set $S$ is given by $\int_{S}\lambda^{k}_{r}(x)g^{k}_{r}(s)dsdx$.

To determine the \textit{mean} number of busy servers $\rho^{k}_{r}(t)$, we evaluate the integral $\int_{S}\lambda^{k}_{r}(x)g^{k}_{r}(s)ds dx$ over the shaded area illustrated in Figure \ref{fig:mginf}. This shaded area represents arrivals to the $\text{M}_{t}/\text{GI}/\infty$ queue since time $kw$ that have not yet completed at time $t$. The resulting expression for $\rho^{k}_{r}(t)$ is given in Equation \ref{eqn:rho2}. If we further consider that the arrival rate $\lambda^{k}_{r}(x)$ is constant over the time window such that $\lambda^{k}_{r}(x)=\lambda^{k}_{r}$, the expression for $\rho^{k}_{r}(t)$ simplifies as shown in Equation \ref{eqn:rho2const}.

Thus, within each window, $N^{k,\infty}_{r}  \left(  \tau_{i} \right)$ is Poisson distributed with a time-dependent mean $\rho^{k}_{r}(\tau_{i})$. Given a specific value $c^{k}_{r}$, we can use this characterization of $N^{k,\infty}_{r}  \left(  \tau_{i} \right)$ to evaluate the upper bound at any time $\tau_{i}$.

\begin{equation}
\begin{aligned}
\rho^{k}_{r}(t)=\int_{kw}^{t} \int_{t-x}^{\infty}\lambda^{k}_{r}(x)g^{k}_{r}(s)dsdx 
\end{aligned}
\label{eqn:rho2}
\end{equation}
\begin{equation}
\begin{aligned}
\rho^{k}_{r}(t)&=\int_{kw}^{t} \int_{t-x}^{\infty}\lambda^{k}_{r}g^{k}_{r}(s)dsdx \\&= \lambda^{k}_{r}\left[t-kw- \int_{0}^{t-kw} G^{k}_{r}(x)dx \right]
\end{aligned}
\label{eqn:rho2const}
\end{equation}

\subsection{Target Predictions for Bounding the Time-Averaged Blocking Probability}
\label{sec:avgblock}
Knowing that we can evaluate the upper bound on the blocking probability at any time and for any $c^{k}_{r}$, we now investigate the minimal value of $c^{k}_{r}$ that limits the \textit{time-averaged} blocking probability to be below a threshold $\delta$. This minimal $c^{k}_{r}$ will be referred to as the \textit{target}, and it represents the number of drivers that the platform seeks to supply during the upcoming time window to limit reach time service violations (i.e., to limit the fraction of non-reserved requests whose reach time will exceed the reach time service requirement).

Precisely, the time-averaged blocking probability in region $r\in R$ during window $\left( kw, (k+1)w\right]$ is given in Equation \ref{eqn:avgBlock1}, where $\gamma_{t}$ is an indicator random variable that takes the value one if a passenger that arrives at time $t$ would be admitted. Since Poisson arrivals see time averages (PASTA property), the time-averaged blocking probability is equivalent to the blocking probability of a typical non-reserved ride request that appears between $\left( kw, (k+1)w\right]$. Then, the target $c^{k}_{r}$ is the desired number of drivers that restricts this time-averaged blocking probability. In other words, the target $c^{k}_{r}$ is the desired number of drivers that limits the blocking probability of a typical non-reserved ride request that will appear during the upcoming window. As previously mentioned, evaluating the blocking probability in Equation \ref{eqn:avgBlock1} is challenging. Thus, to compute the target, we use the time-averaged value of the upper bound in Theorem \ref{thm:ub}. As shown in Inequality \ref{eqn:avgBlock}, if we find the value of $c^{k}_{r}$ that limits the time-averaged upper bound to be less than the threshold $\delta$, then this $c^{k}_{r}$ will also limit the time-averaged blocking probability to be less than $\delta$. 

\begin{equation}
\begin{aligned}
B^{k}_{r} = \frac{1}{w} \int_{kw}^{(k+1)w}P(\gamma_{t}=0) dt
\end{aligned}
\label{eqn:avgBlock1}
\end{equation}

\begin{equation}
\begin{aligned}
B^{k}_{r} \leq \frac{1}{w} \int_{kw}^{(k+1)w}P\left( N^{k,\infty}_{r}  \left(  t \right) \geq  c^{k}_{r} - \displaystyle \max_{  \hat{t}\in\left( t, (k+1)w \right] } \left[ f_{r}^{P,k}(\hat{t}) + f_{r}^{BA,k}(\hat{t}) \right]  \right)dt \leq \delta 
\end{aligned}
\label{eqn:avgBlock}
\end{equation}

Therefore, as shown in Equation \ref{eqn:computeC}, we seek the minimal value $c^{k}_{r}$ that restricts $B^{k}_{r}$ to be less than or equal to the threshold $\delta$. In Equation \ref{eqn:computeC}, observe that the time-averaged upper bound on the blocking probability decreases monotonically with increasing values of $c$; consequently, since $c$ must be a non-negative integer, we can iterate through increasing integer values of $c$ until we find the minimal target $c^{k}_{r}$ that ensures that the time-averaged blocking probability is less than $\delta$ (alternatively, we may use faster line search techniques). Note that just as we can evaluate the upper bound in Theorem \ref{thm:ub} for a specific value of $c$ and at a specific time (Section \ref{sec:transanal}), we can evaluate the time-averaged upper bound for a specific value of $c$ using numerical integration.

\begin{equation}
\begin{aligned}
c^{k}_{r} = &\min_{c\geq 0, \; c \in \mathbb{Z}}  \biggl[ c: \\  &\frac{1}{w} \int_{kw}^{(k+1)w}P\left( N^{k,\infty}_{r}  \left(  t \right) \geq  c - \displaystyle \max_{  \hat{t}\in\left( t, (k+1)w \right] } \left[ f_{r}^{P,k}(\hat{t}) + f_{r}^{BA,k}(\hat{t}) \right]  \right)dt \leq \delta \biggr]
\end{aligned}
\label{eqn:computeC}
\end{equation}

The targets $c^{k}_{r}$ are computed for every region $r \in R$ at the beginning of window $k$ (i.e., at time $t=kw$). If the number of drivers supplied by the platform in each region (either idling in the region or serving requests that initiate in the region) is equal to the corresponding target, then the blocking probability for future non-reserved requests would be less than the threshold $\delta$. Thus, if the targets are provided in each region, the reach time service requirement is probabilistically guaranteed for stochastic non-reserved rides (for book-ahead rides, the reach time service requirement is guaranteed based on the admission control policy in Section \ref{sec:adcont}). Apart from target computations, the upper bound on the blocking probability can be used as a performance measure for the admission control policy, where performance of the policy refers to the probability of reach time service violation (for a given level of driver supply).

\section{Driver Dispatching \& Rebalancing Mechanism}
\label{sec:vehdip}
In this section, we develop a driver dispatching and rebalancing mechanism that aims to maintain the targets across multiple regions. The targets computed in Section \ref{sec:targets} represent a desired level of driver supply such that providing the targets in a region probabilistically guarantees the reach time service requirement for non-reserved ride requests. In practice, within the time window $\left( kw, (k+1)w\right]$, drivers serving requests that initiated in a region $r\in R$ may finish their trips in other regions. Similarly, drivers serving requests that initiated in an external region $r'\in R \backslash \{r\}$ may finish their trip in region $r$. Thus, the number of drivers associated with each region may deviate from the corresponding target $c^{k}_{r}$ due to observed origin-destination trip patterns. This section presents a dispatching/rebalancing mechanism that computes the minimum number of driver transitions that achieve the targets, where only idle drivers are allowed to transition between adjacent regions. We show that the proposed optimization formulation reduces to a \textit{minimum cost flow} formulation on a transformed network of regions.

In more detail, consider that at some time $t$ the platform aims to determine the necessary driver transitions that maintain the targets. In this section, all the defined variables represent the network conditions at time $t$; this time $t$ could be either at the beginning of time window $\left( kw, (k+1)w\right]$ or within the window. For every region $i$, let $a_{i}$ be the number of active drivers serving requests initiated in the region, and let $e_{i}$ be the number of idle drivers in the region. In addition, for every region, define a virtual supply $s^{v}_{i}$ as shown in Equation \ref{eqn:vsup}, where the virtual supply represents the number of excess drivers (beyond the target) that can transition to adjacent regions. The virtual supply $s^{v}_{i}$ is limited by the number of idle drivers in the region; thus, it is the minimum of the idle drivers $e_{i}$ and the number of drivers in excess of the target $\left(a_{i}+e_{i}\right)- c^{k}_{i}$. Similarly, define a virtual demand $d^{v}_{i}$ as shown in Equation \ref{eqn:vdem}, where the virtual demand represents the number of additional drivers needed in region $i$ to meet the target $c^{k}_{i}$ at time $t$. Furthermore, for every region $i$, define $\Delta_{i}$ as shown in Equation \ref{eqn:delta}, where $\Delta_{i}$ represents either the demand (expressed as a negative value) or the supply.

\begin{equation}
\begin{aligned}
s^{v}_{i} = 
\begin{cases}
\min \left\lbrace e_{i}, \left(a_{i}+e_{i}\right)- c^{k}_{i} \right\rbrace & \text{if} \quad c^{k}_{i}-\left( a_{i}+e_{i} \right)\leq 0\\
0 & \text{otherwise}\label{eqn:vsup}
\end{cases}
\end{aligned}
\end{equation}

\begin{equation}
\begin{aligned}
d^{v}_{i} = 
\begin{cases}
c^{k}_{i}-\left( a_{i}+e_{i} \right) & \text{if} \quad c^{k}_{i}-\left( a_{i}+e_{i} \right)>0\\
0 & \text{otherwise}
\end{cases}
\end{aligned}
\label{eqn:vdem}
\end{equation}

\begin{equation}
\begin{aligned}
\Delta_{i} = 
\begin{cases}
-\left[c^{k}_{i}-\left( a_{i}+e_{i} \right)\right] & \text{if} \quad c^{k}_{i}-\left( a_{i}+e_{i} \right)>0\\
\min \left\lbrace e_{i}, \left(a_{i}+e_{i}\right)- c^{k}_{i} \right\rbrace & \text{otherwise}
\end{cases}
\end{aligned}
\label{eqn:delta}
\end{equation}

For the regions defined in Section \ref{sec:sysmod}, we construct a directed network $G=(R,E)$. The set of regions $R$ corresponds to the nodes of the network. The set of edges $E$ includes links $(i,j)$ and $(j,i)$ for every pair of \textit{adjacent} regions $i$ and $j$ (see original network in Figure \ref{fig:transfnet}). Define $h_{ij}$ as the number of drivers that need to transition from region $i$ to the adjacent region $j$ on link $(i,j)$. The platform rebalancing optimization formulation is shown in Equations \ref{eq:obj0}--\ref{eq:cons04}. In this formulation, the platform seeks to minimize the number of driver transitions (objective \ref{eq:obj0}) while ensuring that the targets are maintained (constraint \ref{eq:const01}). In particular, constraint \ref{eq:const01} specifies that the difference between drivers leaving a region and drivers arriving to a region should match the supply/demand in the region. Constraint \ref{eq:const02} restricts the number of drivers leaving a region to the number of idle drivers in the region; in other words, this constraint ensures that the optimal solution to formulation \ref{eq:obj0}--\ref{eq:cons04} (if it exists) describes the number of \textit{idle} drivers transitions to \textit{adjacent} regions (i.e., idle drivers do not transition across multiple regions). The remaining constraints \ref{eq:cons03} and \ref{eq:cons04} ensure that the decision variables $h_{ij}$ are non-negative integers.

\begin{align}
&\min_{ h_{ij}: (i,j)\in E } \qquad \sum_{(i,j)\in E}h_{ij} \label{eq:obj0}\\
&\textrm{s.t.} \quad \sum_{j:(i,j)\in E}h_{ij} - \sum_{j:(j,i)\in E}h_{ji} = \Delta_{i} \qquad \forall i \in R \label{eq:const01}\\
&\quad  \quad   \sum_{j:(i,j)\in E}h_{ij} \leq e_{i} \qquad \forall i \in R \label{eq:const02}\\
&\quad \quad h_{ij} \geq 0 \qquad \forall (i,j)\in E \label{eq:cons03}\\
&\quad \quad h_{ij} \in \mathbb{Z} \qquad \forall (i,j)\in E \label{eq:cons04}
\end{align}   

In formulation \ref{eq:obj0}--\ref{eq:cons04}, unless the total supply matches the total demand ($\sum_{i\in R}s^{v}_{i}=\sum_{i\in R}d^{v}_{i}$) and the network is strongly connected, the optimization problem may not have a feasible solution. Thus, we consider instead the revised formulation \ref{eq:obj1}--\ref{eq:cons15}, where $h_{i}$ corresponds to drivers added/removed from region $i$ by adjusting the total number of drivers in the network. Since adding or removing drivers would be costly to the platform (e.g., requires incentivizing new drivers or taking drivers offline), we associate a high cost $M$ with such transitions. As a result, in the optimal solution to formulation \ref{eq:obj1}--\ref{eq:cons15}, the total number of drivers is adjusted only if the targets could not be maintained internally via transitions of idle drivers across adjacent regions.

\begin{align}
&\min_{ h_{ij}: (i,j)\in E,\; h_{i}: i\in R } \qquad \sum_{(i,j)\in E}h_{ij} + M\sum_{i\in R}|h_{i}|\label{eq:obj1}\\
&\textrm{s.t.} \quad \sum_{j:(i,j)\in E}h_{ij} - \sum_{j:(j,i)\in E}h_{ji} + h_{i} = \Delta_{i} \qquad \forall i \in R \label{eq:const11}\\
&\quad  \quad   \sum_{j:(i,j)\in E}h_{ij} \leq e_{i} \qquad \forall i \in R \label{eq:const12}\\
&\quad \quad h_{ij} \geq 0 \qquad \forall (i,j)\in E \label{eq:cons13}\\
&\quad \quad h_{ij} \in \mathbb{Z} \qquad \forall (i,j)\in E \label{eq:cons14}\\
&\quad \quad h_{i} \in \mathbb{Z} \qquad \forall i\in R \label{eq:cons15}
\end{align}   

Let $h_{i\bullet}$ and $h_{\bullet i}$ be defined as in Equations \ref{eqn:hbul1} and \ref{eqn:hbul2}. In this case, $h_{\bullet i}$ corresponds to drivers added to region $i \in R$ by adjusting the total number of drivers, and $h_{i\bullet}$ corresponds to drivers removed from region $i \in R$ by adjusting the total number of drivers (i.e., $h_{i\bullet}$ represents drivers that can be removed from the system to avoid having excess idle drivers).

\begin{equation}
	\begin{aligned}
		h_{i\bullet} = 
		\begin{cases}
			h_{i} & \text{if} \quad h_{i}>0\\
			0 & \text{otherwise}
		\end{cases}
	\end{aligned}
	\label{eqn:hbul1}
\end{equation}

\begin{equation}
	\begin{aligned}
		h_{\bullet i} = 
		\begin{cases}
			|h_{i}| & \text{if} \quad h_{i}<0\\
			0 & \text{otherwise}
		\end{cases}
	\end{aligned}
	\label{eqn:hbul2}
\end{equation}

Moreover, for notational convenience in mapping the problem to a min-cost flow reformulation, define for each region $i \in R$ variables $h_{ii^{\star}}$ that represent the total number of drivers leaving region $i$ to adjacent regions (Equation \ref{eq:shiistar}). In addition, for each link $(i,j)\in E$, define variables $h_{i^{\star}j}=h_{ij}$. Thus, we can define $h_{ii^{\star}}$ in terms of $h_{i^{\star}j}$ as in Equation \ref{eq:shiistar2}. Since $h_{ij}$ is a non-negative integer for all $(i,j)\in E$, we have that $h_{ii^{\star}}$ and $h_{i^{\star}j}$ are non-negative integers as well.
\begin{align}
	h_{ii^{\star}} &= \sum_{j:(i,j)\in E}h_{ij} \qquad \forall i \in R \label{eq:shiistar}\\
	&=\sum_{j:(i,j)\in E}h_{i^{\star}j} \qquad \forall i \in R \label{eq:shiistar2}
\end{align}   

In Appendix B, through a sequence of reformulations, we show that optimization problem \ref{eq:obj1}--\ref{eq:cons15} reduces to the formulation \ref{eq:obj4}--\ref{eq:cons46ex}.

\begin{align}
&\min_{ h_{i^{\star}j}: (i,j)\in E,\; h_{i \bullet}, h_{\bullet i}, h_{ii^{\star}}: i\in R,\;\bar{h} } \qquad \sum_{i\in R}h_{ii^{\star}} + M\sum_{i\in R} \left[ h_{i \bullet} +  h_{\bullet i}\right] \label{eq:obj4}\\
&\textrm{s.t.} \quad h_{ii^{\star}} - \sum_{j:(j,i)\in E}h_{j^{\star}i} + h_{i \bullet} - h_{\bullet i} = \Delta_{i} \qquad \forall i \in R \label{eq:const41}\\
&\quad  \quad \sum_{i\in R}h_{\bullet i} + \bar{h} = \sum_{i\in R}d^{v}_{i} \label{eq:const42ex1}\\
&\quad  \quad -\left[\sum_{i\in R}h_{i \bullet} + \bar{h}\right] = -\sum_{i\in R}s^{v}_{i} \label{eq:const42ex2}\\
&\quad  \quad \sum_{j:(i,j)\in E}h_{i^{\star}j} - h_{ii^{\star}} = 0 \qquad \forall i\in R \label{eq:const42ex3}\\
&\quad \quad 0 \leq h_{ii^{\star}} \leq e_{i} \qquad \forall i\in R \label{eq:cons43ext}\\
&\quad \quad h_{i^{\star}j} \geq 0 \qquad \forall (i,j)\in E \label{eq:cons43}\\
&\quad \quad  h_{i \bullet}, h_{\bullet i} \geq 0 \qquad \forall i\in R \label{eq:cons44}\\
&\quad \quad  \bar{h} \geq 0  \label{eq:cons44ex}\\
&\quad \quad h_{ii^{\star}} \in \mathbb{Z} \qquad \forall i\in R \label{eq:cons46extra}\\
&\quad \quad h_{i^{\star}j} \in \mathbb{Z} \qquad \forall (i,j)\in E \label{eq:cons45}\\
&\quad \quad h_{i \bullet}, h_{\bullet i} \in \mathbb{Z} \qquad \forall i\in R \label{eq:cons46}\\
&\quad \quad  \bar{h} \in \mathbb{Z} \label{eq:cons46ex}
\end{align}

Consider the standard minimum cost flow problem given in formulation \ref{eq:objmcf}--\ref{eq:constmcf2} for a network $G'=(V,A)$ \citep{ahuja1993,wolsey1998}, where $c_{pq}$ is the cost of a unit flow on link $(p,q)\in A$, $x_{pq}$ are decision variables corresponding to flows on each link $(p,q)\in A$, $b_{p}$ is the equivalent of supply/demand at node $p$, and $u_{pq}$ is an upper bound on the flows $x_{pq}$ (i.e., capacity of link $(p,q)\in A$). A necessary condition for feasibility of the optimization problem is $\sum_{p\in V}b_{p}=0$.
\begin{align}
&\min_{x_{pq}:(p,q)\in A} \qquad \sum_{(p,q)\in A}c_{pq}x_{pq} \label{eq:objmcf}\\
&\textrm{s.t.} \quad \sum_{ \{q:(p,q)\in A \}   }x_{pq} - \sum_{ \{q:(q,p)\in A \}   }x_{qp} = b_{p} \qquad \forall p\in V \label{eq:constmcf1}\\
&\quad  \quad 0 \leq x_{pq} \leq u_{pq} \qquad \forall  (p,q) \in A \label{eq:constmcf2}
\end{align}  

Apart from the integrality constraints, the formulation \ref{eq:obj4}--\ref{eq:cons46ex} has the same structure as the minimum cost flow optimization problem \ref{eq:objmcf}--\ref{eq:constmcf2}; this implies that the constraint matrix associated with formulation \ref{eq:obj4}--\ref{eq:cons46ex} is totally unimodular. Thus, since $\Delta_{i}$, $d^{v}_{i}$, $s^{v}_{i}$, and $e_{i}$ are all integer values, each extreme point in the constraint set will be integral. Then, solving the linear programming relaxation in \ref{eq:obj5}--\ref{eq:cons54ex} will give us the \textit{integer optimal solution} of optimization problem \ref{eq:obj4}--\ref{eq:cons46ex}.

\begin{align}
&\min_{ h_{i^{\star}j}: (i,j)\in E,\; h_{i \bullet}, h_{\bullet i}, h_{ii^{\star}}: i\in R,\;\bar{h} } \qquad \sum_{i\in R}h_{ii^{\star}} + M\sum_{i\in R} \left[ h_{i \bullet} +  h_{\bullet i}\right] \label{eq:obj5}\\
&\textrm{s.t.} \quad h_{ii^{\star}} - \sum_{j:(j,i)\in E}h_{j^{\star}i} + h_{i \bullet} - h_{\bullet i} = \Delta_{i} \qquad \forall i \in R \label{eq:const51}\\
&\quad  \quad \sum_{i\in R}h_{\bullet i} + \bar{h} = \sum_{i\in R}d^{v}_{i} \label{eq:const52ex1}\\
&\quad  \quad -\left[\sum_{i\in R}h_{i \bullet} + \bar{h}\right] = -\sum_{i\in R}s^{v}_{i} \label{eq:const52ex2}\\
&\quad  \quad \sum_{j:(i,j)\in E}h_{i^{\star}j} - h_{ii^{\star}} = 0 \qquad \forall i\in R \label{eq:const52ex3}\\
&\quad \quad 0 \leq h_{ii^{\star}} \leq e_{i} \qquad \forall i\in R \label{eq:cons53ext}\\
&\quad \quad h_{i^{\star}j} \geq 0 \qquad \forall (i,j)\in E \label{eq:cons53}\\
&\quad \quad  h_{i \bullet}, h_{\bullet i} \geq 0 \qquad \forall i\in R \label{eq:cons54}\\
&\quad \quad  \bar{h} \geq 0  \label{eq:cons54ex}
\end{align}

The linear program \ref{eq:obj5}--\ref{eq:cons54ex} can be mapped to a minimum cost flow program \ref{eq:objmcf}--\ref{eq:constmcf2} applied on a transformed network illustrated in Figure \ref{fig:transfnet}. In particular, consider a source node $\text{SO}$ where links $\left(\text{SO},i\right)$ that connect $\text{SO}$ to region $i\in R$ dispatch flows $h_{\bullet i}$. In addition, consider a sink node $\text{SI}$ where links $\left(i,\text{SI}\right)$ that connect region $i\in R$ to $\text{SI}$ dispatch flows $h_{i \bullet}$. Let $\bar{h}$ represent the flow between $\text{SO}$ and $\text{SI}$. Then, observe that constraint \ref{eq:const51} is equivalent to constraint \ref{eq:constmcf1} at all un-starred nodes in the network transformation of Figure \ref{fig:transfnet}. Similarly, constraint \ref{eq:const52ex3} is equivalent to constraint \ref{eq:constmcf1} at all starred nodes. Constraint \ref{eq:const52ex1} corresponds to constraint \ref{eq:constmcf1} applied at the source node $\text{SO}$, and constraint \ref{eq:const52ex2} corresponds to constraint \ref{eq:constmcf1} applied at the sink node $\text{SI}$. In the network transformation, each link is associated with a $\left(\text{cost} , \text{capacity} \right)$ label. Observe that the objective function \ref{eq:obj5} can be obtained by plugging the link costs and flow variables in the minimum cost flow objective function \ref{eq:objmcf}. Also, observe that constraints \ref{eq:cons53ext}--\ref{eq:cons54ex} are the link capacity constraints \ref{eq:constmcf2} in the transformed network. Furthermore, by definition, $\sum_{i\in R}\Delta_{i}+\sum_{i\in R}d^{v}_{i}-\sum_{i\in R}s^{v}_{i}=0$; this implies that the necessary condition for feasibility in the minimum cost flow program ($\sum_{p\in V}b_{p}=0$) is satisfied. Thus, solving the linear program \ref{eq:obj5}--\ref{eq:cons54ex} is equivalent to solving the minimum cost flow program \ref{eq:objmcf}--\ref{eq:constmcf2} using the transformed network.

\begin{figure}[H]
	\centerline{\includegraphics[width=0.56\textwidth]{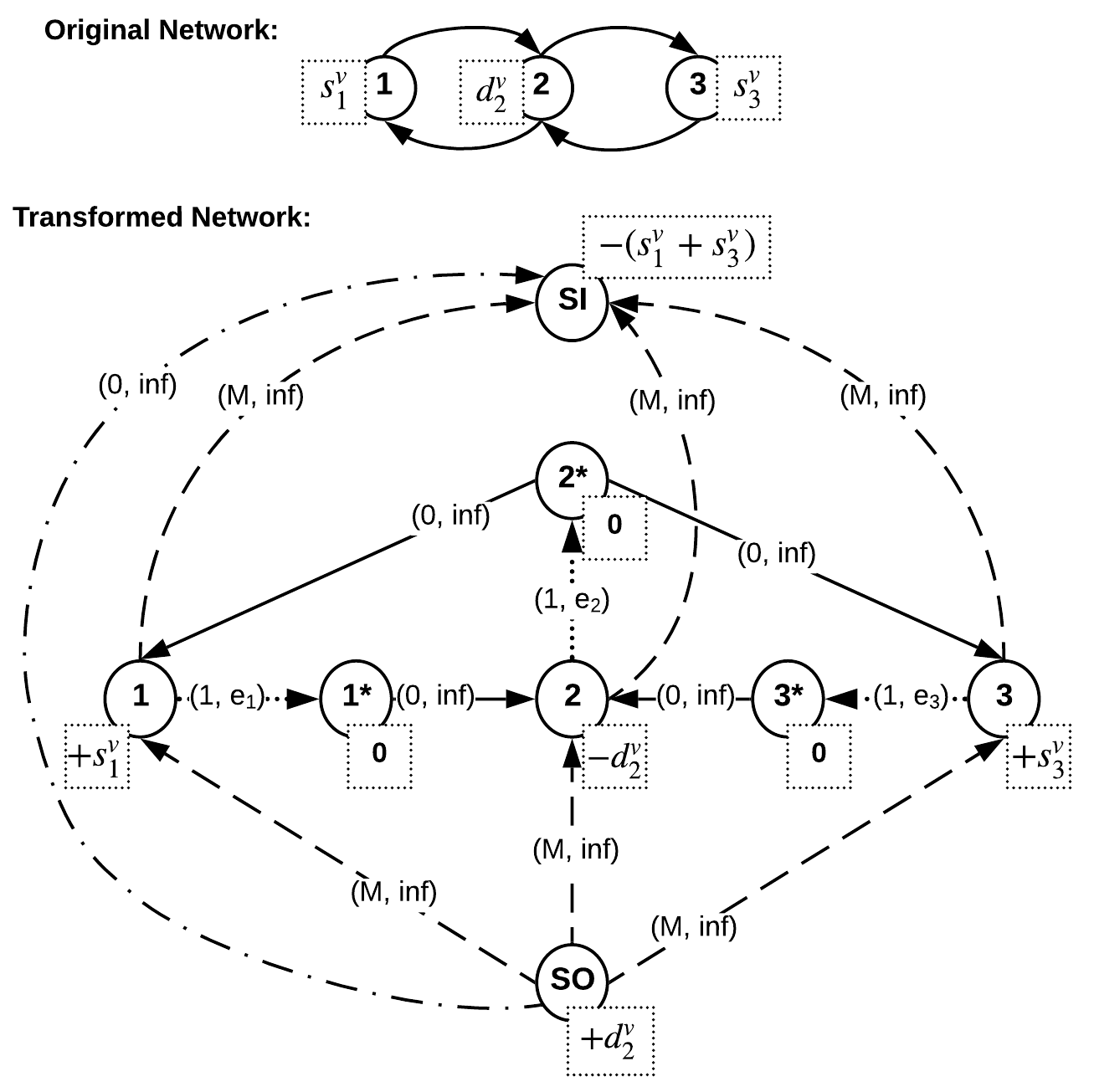}}
	\caption{Network transformation corresponding to the minimum cost flow program, where solving the integer program \ref{eq:obj1}--\ref{eq:cons15} using the original network is equivalent to solving the minimum cost flow program \ref{eq:objmcf}--\ref{eq:constmcf2} using the transformed network. Each link in the transformed network is associated with a $\left(\text{cost} , \text{capacity} \right)$ label. Each node in the transformed network is either a supply, demand, or transmission node such that values of $b_{p}$ in constraint \ref{eq:constmcf1} are within the squares.}
	\label{fig:transfnet}
\end{figure}

Consequently, since the integer program \ref{eq:obj1}--\ref{eq:cons15} reduces to formulation \ref{eq:obj5}--\ref{eq:cons54ex}, then solving the integer program \ref{eq:obj1}--\ref{eq:cons15} on the original network (Figure \ref{fig:transfnet}) is equivalent to solving the minimum cost flow program \ref{eq:objmcf}--\ref{eq:constmcf2} on the illustrated transformed network. As a minimum cost flow program, the driver dispatching and rebalancing optimization problem can be solved in polynomial time. The optimal solution of the optimization program represents recommended idle driver transitions that are needed to maintain the targets across regions. Specifically, the optimal solution includes idle drivers that should transition to adjacent regions \textit{and} idle drivers that should be added to the network by adjusting the total number of drivers in the system. In addition, the optimal solution also includes excess idle drivers that can be removed from the system.

\section{Simulation Results}
\label{sec:manhsim}

In this section, we present experimental results using data from Lyft operations in Manhattan, NYC on Friday December 14th, 2018 \citep{nyctlc2019}. We consider trips that started between 16:00--19:00 (local time) in four regions. The regions chosen roughly correspond to four sections of the city as illustrated in Figure \ref{fig:mhtn} (1-lower Manhattan, 2-midtown Manhattan, 3-upper west side, and 4-upper east side). For time windows of duration $w=20$ minutes, we use trip initiation and completion time data available on the New York City Taxi and Limousine Commission website to characterize the processes $\{f_{r}^{P,k}(t), f_{r}^{BA,k}(t), N^{k}_{r}(t):t\in \left( kw, (k+1)w\right] \}$. Our primary findings suggest that an increase in the fraction of book-ahead rides leads to a reduction in the total number of drivers that are needed to probabilistically guarantee the reach time service requirement. This reduction in the total number of drivers is also associated with a lower number of idling drivers (i.e., an increase in the driver utilization rate).

\begin{figure}
	\centerline{\includegraphics[width=0.4\textwidth]{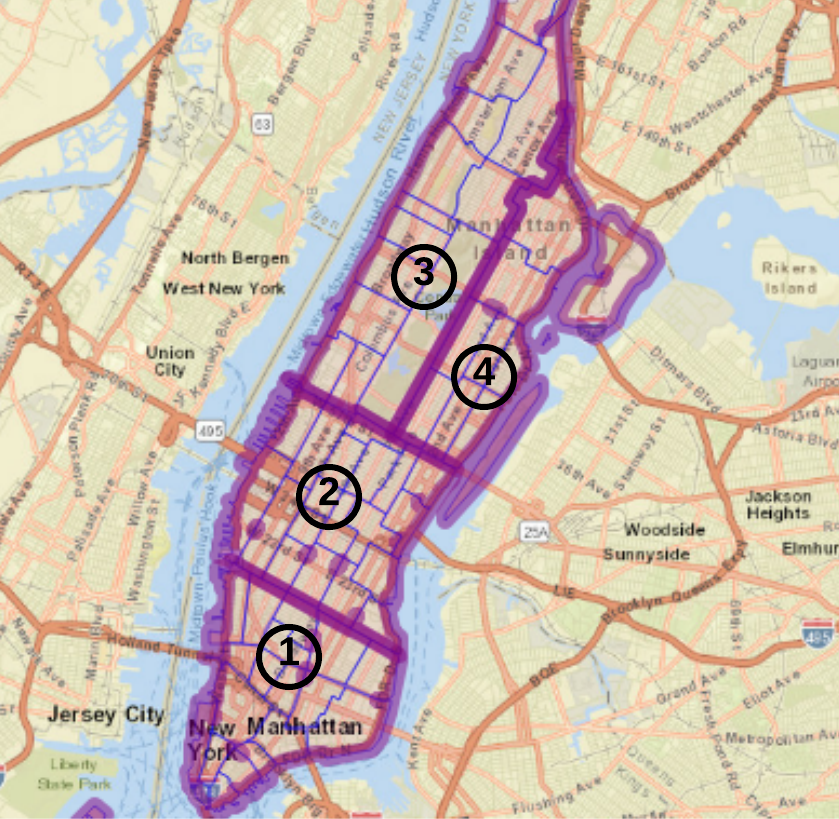}}
	\caption{Manhattan divided into four regions}
	\label{fig:mhtn}
\end{figure}

\subsection{System model specification and comparison to observed data}
\label{sec:simsystemspec}
The process $\{f_{r}^{P,k}(t):t\in \left( kw, (k+1)w\right] \}$ is generated at the beginning of every window $k$. Specifically, using the available data, $f_{r}^{P,k}(t)$ represents \textit{previously observed} rides that initiated in region $r$ prior to $t=kw$ and will be active at time $t\in \left( kw, (k+1)w\right]$.

To generate the process $\{f_{r}^{BA,k}(t):t\in\left( kw, (k+1)w\right]\}$ from the New York City data, we randomly sample a fraction $p_{BA}$ of the trips that start during window $k$ in region $r$. We choose to generate $f_{r}^{BA,k}(t)$ as the fraction of anticipated rides since we are interested in analyzing the change in the target number of drivers as the fraction of book-ahead rides increases.

As for the stochastic process $\{N^{k}_{r}(t):t\in \left( kw, (k+1)w\right] \}$, at the beginning of each window $k$, we calibrate the demand rate $\lambda^{k}_{r}$ corresponding to ride requests that will appear during the upcoming window in region $r$. In the following simulation, for simplicity, the demand rate varies across time-windows but is assumed constant within each time window; however, the proposed framework can be implemented using time-dependent demand rate functions by evaluating Equation \ref{eqn:rho2}. Moreover, even with window-constant demand rates, the Poisson distribution describing active drivers is time-varying within each window such that the mean is given by Equation \ref{eqn:rho2const}. We emphasize that this transient analysis does not assume an equilibrium or steady-state conditions in any time window. The arrival rate for region 2 is shown in Figure \ref{fig:arrivalrates}; as observed, the demand rate increases rapidly showing the need for non-equilibrium methods. For the distribution $g^{k}_{r}(\cdot)$ representing ride duration, we use the empirical distribution that is derived from the observed rides in each region. Note that to analyze the change in the target number of drivers with increasing book-ahead rides, we effectively assume that the arrival rate of non-reserved ride requests is $(1-p_{BA})\lambda^{k}_{r}
$ (where a fraction $p_{BA}$ of the anticipated trips that will initiate during window $k$ are book-ahead rides). 

\begin{figure}
	\centerline{\includegraphics[width=0.6\textwidth]{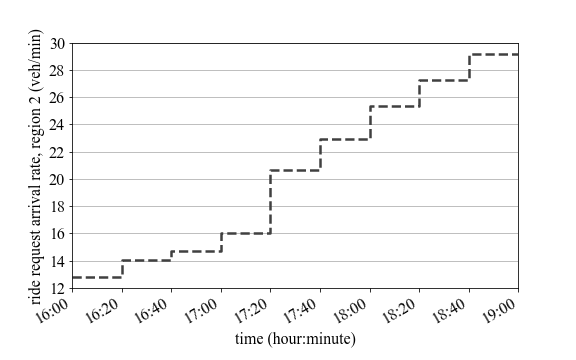}}
	\caption{Arrival rate for ride requests that initiate in region 2.}
	\label{fig:arrivalrates}
\end{figure}

As illustrated in Figure \ref{fig:modver}, the proposed model for predicting the number of active rides (Section \ref{sec:sysmod}) accurately represents the observed data. In this Figure, for comparison with observed trip data, we consider that all rides are admitted and that there are no book-ahead rides (effectively assuming  $N^{k}_{r}(t)=N^{k,\infty}_{r}  \left(  t \right)$). Recall that $N^{k}_{r}(t)$ represents the \textit{predicted} non-reserved ride requests that will appear during window $k$; in contrast, during window $k+1$, the process $f_{r}^{P,k+1}(t)$ consists of observed trips (as given in the data) that differ from the previously predicted trips.

\begin{figure}[H]
	\centerline{\includegraphics[width=0.6\textwidth]{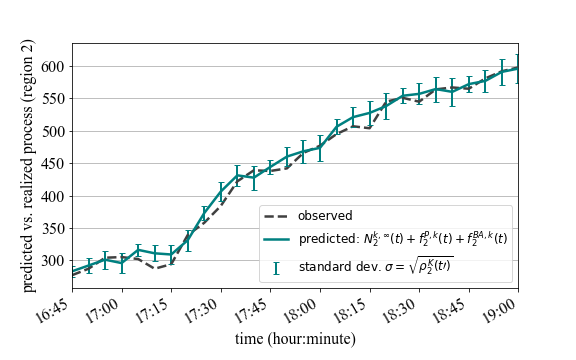}}
	\caption{Predicted total number of active rides vs. observed number of active rides, where predictions were made over time windows with a duration of $20$ minutes. The error bars correspond to one standard deviation of the time-dependent Poisson distribution characterizing $N^{k,\infty}_{r}$. In this figure, to compare with the observed trip data, we assume that all rides are admitted (i.e., we consider that $N^{k}_{r}(t)=N^{k,\infty}_{r}  \left(  t \right)$).}
	\label{fig:modver}
\end{figure}

\subsection{Upper bound on the blocking probability}
To evaluate how tight is the upper bound in Inequality 22, we implement the admission control policy in region 2 and average the observed proportion of blocked rides $B_{r}^{k}$ across time windows. For this upper bound numerical analysis, the assumptions involved in target evaluation and admission control apply; specifically, the total supply (active and idle) is maintained at the target level, drivers switch between active and idle within the region, and non-reserved rides are blocked if upon admission the total number of active rides would exceed the target at some point in time throughout the ride duration. Figure 9 shows the variation in the blocking proportion $B_{r}^{k}$ relative to the upper bound $\delta$. As observed, the blocking proportion $B_{r}^{k}$ increases with larger tolerance values. We also observe that the blocking proportion increases with the fraction of book-ahead rides $p_{BA}$ as a result of fewer idle drivers being available for non-reserved rides.

\begin{figure}[H]
	\centerline{\includegraphics[width=0.6\textwidth]{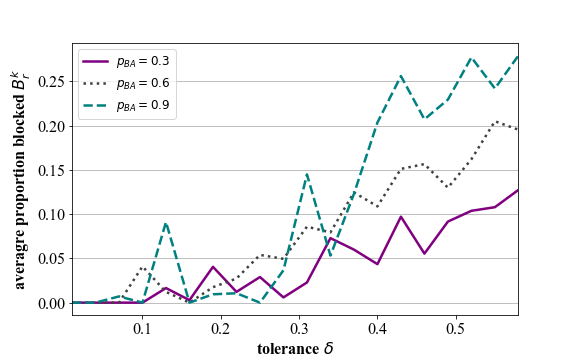}\includegraphics[width=0.6\textwidth]{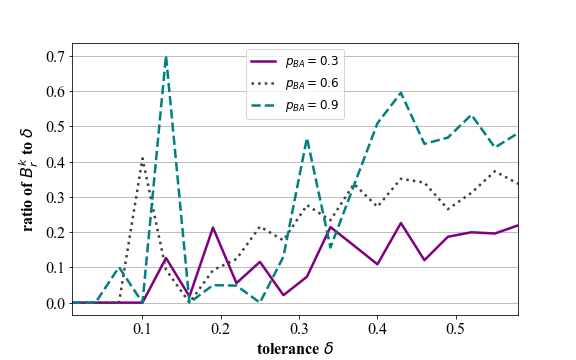}}
	\caption{The change in observed blocking proportion $B_{r}^{k}$ and the ratio $B_{r}^{k}$/$\delta$ relative to the upper bound $\delta$. }
	\label{fig:upper_bound}
\end{figure}

\subsection{Target computations, admission control, and minimum cost flow dispatching/rebalancing}
Then, to account for the spatial distribution of demand and the variation in supply across regions, we implement the proposed framework in Sections \ref{sec:sysmod}--\ref{sec:vehdip} (see Figure \ref{fig:rel}). In particular, we demand moves between regions and that the supply deviates from the target, and we implement the min. cost flow to maintain the target.

First, as mentioned in Section \ref{sec:simsystemspec}, we characterize the processes $\{f_{r}^{P,k}(t), f_{r}^{BA,k}(t), N^{k}_{r}(t):t\in \left( kw, (k+1)w\right] \}$ representing the predicted number of active rides in each region $r$. Then, using the upper bound on the time-dependent blocking probability of the admission control policy, we determine the target number of drivers associated with every region $r$ during the upcoming window. After that, at the beginning of the time window, we apply the driver dispatching/rebalancing mechanism to attain the targets across regions. Then, throughout the time window, for every non-reserved ride request that is received, we implement the admission control policy to determine whether the request should be admitted or blocked; the received non-reserved ride requests are directly retrieved from the New York City data (as opposed to the predictions $N^{k}_{r}(t)$). We also implement the driver dispatching/rebalancing mechanism halfway through the time window. However, at the beginning of the time window we allow for total adjustments of the driver supply while halfway through the window we consider that only existing idle drivers can transition across adjacent regions. This process is then repeated for every time window.

For simulation purposes, we disregard the stochasticity of drivers entering and exiting the system across time windows. However, the admission control policy, target computations, and subsequent driver dispatching policy allow for a time-varying and stochastic variation in the supply that is joining or leaving the platform. In fact, target evaluation is based on the demand process and the admission control assumes that the target is maintained throughout the time window. Even if the actual supply deviates from the target, the admission control policy is still implemented by finding if there are any idle drivers and measuring the change in idle drivers relative to the target. On the other hand, the driver dispatching is only concerned with the instantaneous state of the supply relative to the target.

Note that the presented driver rebalancing strategy only uses information from the current time window. In other words, while the proposed state-dependent strategy does not assume steady-state conditions in a time-varying environment, it does not look into future windows to determine the current rebalancing recommendations. Alternative policies that predict future dynamics multiple windows in advance may also be effective since they would have more information on the anticipated variation in driver supply.

\begin{figure}
	\centerline{\includegraphics[width=0.6\textwidth]{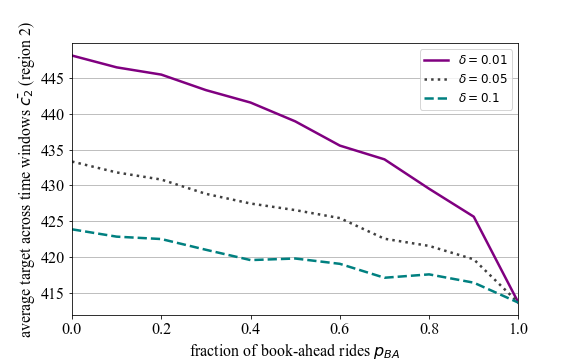}}
	\caption{Change in the time-averaged target number of drivers with an increase in the fraction of book-ahead rides (for different quality of service thresholds $\delta$). For each data point (i.e., every $\left(p_{BA},\delta\right)$ pair), the plotted time-averaged target is the average of the corresponding value obtained from 30 different iterations of the proposed framework, where this averaging is needed due to the randomness in generation of the book-ahead profile $f_{r}^{BA,k}(t)$. }
	\label{fig:target}
\end{figure}
We apply the same framework for different fractions of book-ahead rides and record the target $c^{k}_{r}$ across windows. In Figure \ref{fig:target}, we illustrate the change in targets for different fractions of book-ahead rides. In particular, we measure the time-averaged target $\bar{c}_{r}$ for increasing values of $p_{BA}$ and different quality of service thresholds $\delta$ (as defined in Section \ref{sec:avgblock}, $\delta$ bounds the time-averaged blocking probability such that a lower value of $\delta$ indicates a higher quality of service). As expected, we observe that the target number of drivers increases with decreasing $\delta$; this result implies that a larger number of drivers is needed to guarantee the reach time service requirement for a greater fraction of non-reserved ride requests. We also observe that the target number of drivers decreases as the fraction of book-ahead rides increases. The decrease in targets indicates that the number of drivers needed decreases with more information on anticipated trips. 

For the simulation setting, the ratio of internal driver transitions $\sum_{(i,j)\in E}h_{ij}$ to the total flows ($\sum_{(i,j)\in E}h_{ij}$ + $\sum_{i \in R}|h_{i}|$) was approximately 0.5 when averaged across min-cost flow evaluations. The recommended external flows reflect the additional drivers needed to satisfy increasing demand (Figure \ref{fig:arrivalrates}). This ratio depends on the demand rates, frequency of driver rebalancing, and the spatial distribution of regions. All these parameters would vary between different areas and time periods.

\begin{figure}
	\centerline{\includegraphics[width=0.6\textwidth]{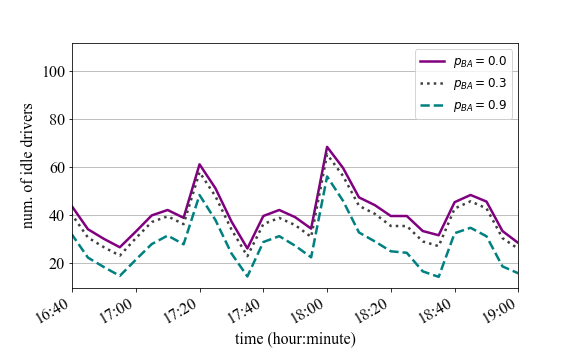}\includegraphics[width=0.6\textwidth]{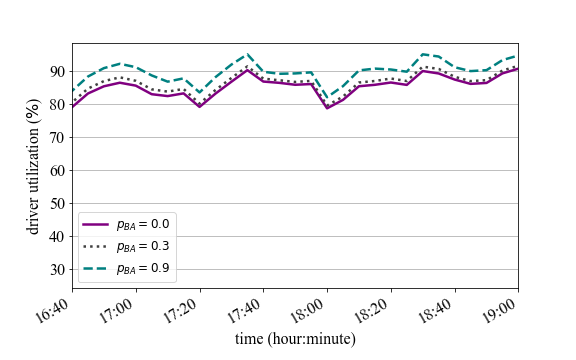}}
	\caption{The number of idle drivers and the driver utilization rate 100*(active/(active+idle)) averaged across regions. The quality of service threshold $\delta$ is set at 0.01.}
	\label{fig:idledr}
\end{figure}

As the target decreases with increasing fractions of book-ahead rides, the number of idling drivers in the system also decreases. Figure \ref{fig:idledr} illustrates the average number of idling drivers for different reservation levels. We observe that when $p_{BA}=0.9$ the average number of idle drivers can be up to $17.3$ less than the corresponding value when $p_{BA}=0.0$. This reduction in the number of idle drivers with increasing $p_{BA}$ translates to a higher driver utilization rate.

Figure \ref{fig:block} illustrates the average number of rides that are blocked by the admission control policy (i.e., the reach time service requirement was not met for these rides). As shown, the average number of blocked rides increases with reservation levels. This increase in blocking results from the reduction in the overall number of drivers in the system. However, the fraction of blocked requests is (mostly) within the specified threshold $\delta=0.01$. For $p_{BA}=0.9$, the fraction of blocked requests slightly exceeds the level of service threshold $\delta$; this discrepancy may be attributed to the randomness in the system and the fact that the targets are not perfectly maintained throughout the entire time window.

\begin{figure}[H]
	\centerline{\includegraphics[width=0.6\textwidth]{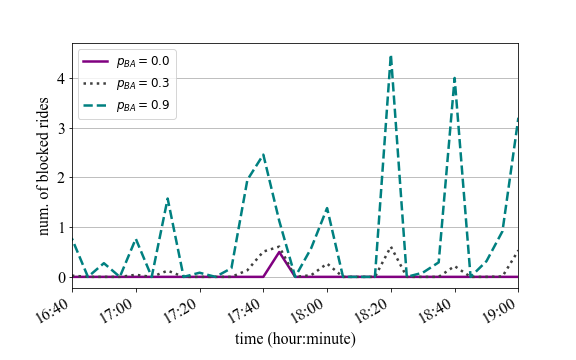}\includegraphics[width=0.6\textwidth]{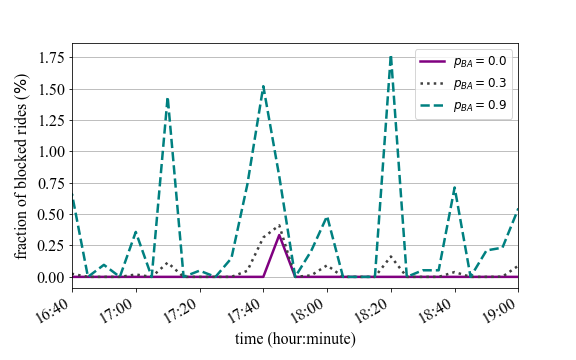}}
	\caption{The number of blocked ride requests and the fraction of blocked requests 100*(blocked/(admitted+blocked)) averaged across regions. The quality of service threshold $\delta$ is set at 0.01.}
	\label{fig:block}
\end{figure}

The previous analysis assumed perfect compliance with inter-regional driver transitions at the simulation-specific driver rebalancing stages (beginning and mid-window). However, the drivers may not follow platform recommendations and that would result in greater difficulty maintaining the targets. Figure \ref{fig:block-no-transitions} shows the number of blocked rides and fraction of blocked rides in the worst-case scenario where drivers do not follow inter-regional transition recommendations. As observed, the number of blocked rides almost doubles in some cases and the fraction of blocked rides also increases up to 3.5\%.

\begin{figure}[H]
	\centerline{\includegraphics[width=0.6\textwidth]{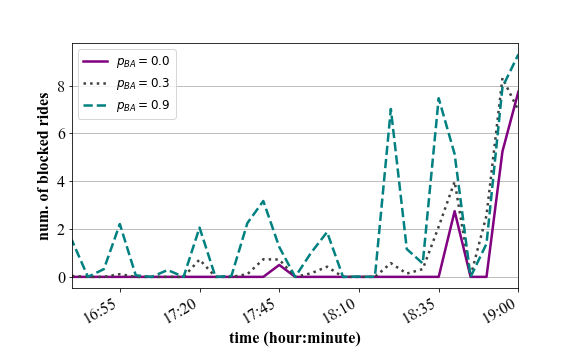}\includegraphics[width=0.6\textwidth]{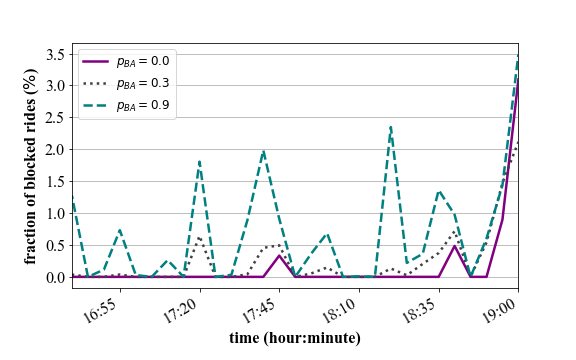}}
	\caption{For the case when idle drivers do not follow platform-recommended transitions between regions, we observe an increase in the number blocked rides and the fraction of blocked rides. The quality of service threshold $\delta$ is set at 0.01.}
	\label{fig:block-no-transitions}
\end{figure}

\section{Conclusion}
\label{sec:conc}

In this article, we propose a model for transient analysis of stochasticity in ridesourcing systems. As opposed to steady-state equilibrium methods, we characterize the time-dependent state of the system and design control policies for managing driver supply. Furthermore, we incorporate book-ahead rides (reservations) in our framework and analyze the impact of book-ahead rides on driver supply management.

In more detail, we propose a state-dependent control policy that assigns drivers to observed ride requests with the objective of guaranteeing the reach time service requirement for book-ahead rides. Then, we derive a time-dependent upper bound on the performance of the control policy, where the performance of the policy is measured in terms of the probability of reach time service violations for non-reserved rides. Subsequently, this upper bound is used to determine the target number of drivers that probabilistically guarantees the reach time service requirement for non-reserved rides. The targets represent the total number of drivers that are associated with a region such that the drivers are either idling in the region or serving requests that initiate in the region. Then, considering a set of regions with different targets, we propose a driver dispatching/rebalancing optimization program that seeks to maintain the targets across regions. We show that the dispatching/rebalancing problem reduces to a minimum cost flow program that is solved on a transformed network.

The key findings are as follows: (1) For the desired reach time quality of service, an increase in the fraction of book-ahead rides leads to a reduction in the total number of drivers required. (2) This reduction in the total number of drivers is associated with a decrease in the number of idling drivers. (3) Once the driver supply is decreased, there is a greater risk that the reach time service requirement will be violated for anticipated non-reserved rides. However, the fraction of rides that experience increased reach time beyond the reach time service requirement is within a specified threshold, where this threshold dictates the target number of required drivers. (4) For Lyft rides in Manhattan, we observe rapid variations in demand rates that emphasize the need for transient analysis of ridesourcing dynamics.

The proposed model can be used for operation of ridesourcing systems. Specifically, the proposed control policy can be used for ensuring reach time priority for book-ahead rides, the target supply determines the number of drivers that would probabilistically guarantee the reach time service requirement for non-reserved rides, and the minimum cost flow program determines the necessary driver dispatching/rebalancing that is needed to maintain the targets.

More importantly, the proposed model can inform policy decisions that seek to maximize driver welfare and to reduce congestion externalities associated with ridesourcing platforms. In particular, for a given quality of service and reach time service requirement, policy makers can determine if the ridesourcing platform is employing an excessive number of drivers by comparing the total number of drivers in the system to the target supply. In addition, our results suggest that policy makers should advocate for an increased fraction of book-ahead rides and supply management strategies that use this book-ahead information to reduce the number of idling drivers.

\section*{Supplementary Material}
Data and code used to generate results in this article are available on\\ \url{https://github.com/spartalab/book-ahead/}

\section*{Acknowledgments}
This research was supported by the National Science Foundation under Grants No. 1562291 and 1826320. Partial support was provided by the Data-Supported Transportation Operations and Planning center (D-STOP). The authors would like to thank Ahmad AlAmmouri for his insightful comments and suggestions.  

\section*{Appendix}
\subsection*{A. Theorem 1 Proof:}
\begin{appendixthm}
	The blocking probability, $P(\gamma_{i}=0)$, for the $i^{\text{th}}$ stochastic non-reserved ride request that appears at time $\tau_{i}$ is bounded above by $P\left( N^{k,\infty}_{r}  \left(  \tau_{i} \right) \geq  c^{k}_{r} - \displaystyle \max_{  t\in\left( \tau_{i}, (k+1)w \right] } \left[ f_{r}^{P,k}(t) + f_{r}^{BA,k}(t) \right]  \right)   $
\end{appendixthm}
\begin{proof}
	We first start by deriving upper bounds on the blocking probability $P(\gamma_{i}=0)$ (Inequalities \ref{eqn:p2}--\ref{eqn:p4}). Then, through Equations \ref{eqn:p5}--\ref{eqn:p9},  we show that the upper bound in Inequality \ref{eqn:p4} can be expressed in terms $N^{k,\infty}_{r}  \left(  \tau_{i} \right)$, where $N^{k,\infty}_{r}  \left(  \tau_{i} \right)$ is the number of busy servers at time $\tau_{i}$ in a transient $\text{M}_{t}/\text{GI}/\infty$ queue that starts empty at the beginning of the time window.
	\begin{flalign}
		&P(\gamma_{i}=0)&\\ 
		&=P\left( \exists t\in \left(\tau_{i}, \min\{\tau_{i}+D_{i}, (k+1)w \}  \right]: 1+ f_{r}^{P,k}(t) + f_{r}^{BA,k}(t) + \sum_{n=1}^{i-1} \mathbf{1}\{ \tau_{n}+D_{n}>t \}\gamma_{n} > c^{k}_{r} \right)&\\
		&\leq P\left( \exists t\in \left(\tau_{i}, \min\{\tau_{i}+D_{i}, (k+1)w \}  \right]: 1+ f_{r}^{P,k}(t) + f_{r}^{BA,k}(t) + \sum_{n=1}^{i-1} \mathbf{1}\{ \tau_{n}+D_{n}>t \} > c^{k}_{r} \right)\label{eqn:p2}&  \\
		&\leq P\left( \exists t\in \left(\tau_{i},  (k+1)w \right]: 1+ f_{r}^{P,k}(t) + f_{r}^{BA,k}(t) + \sum_{n=1}^{i-1} \mathbf{1}\{ \tau_{n}+D_{n}>t \} > c^{k}_{r} \right)\label{eqn:p3}&\\
		&\leq P\left( \exists t\in \left(\tau_{i},  (k+1)w \right]: 1+ f_{r}^{P,k}(t) + f_{r}^{BA,k}(t) + \sum_{n=1}^{i-1} \mathbf{1}\{ \tau_{n}+D_{n}>\tau_{i} \} > c^{k}_{r} \right)\label{eqn:p4}&
	\end{flalign}
	Inequality \ref{eqn:p2} holds since we are considering that all requests that are received before the $i^{\text{th}}$ request are admitted (i.e, $\gamma_{n}=1$ for all $n\in \{1,..., i-1\}$).\\
	Inequality \ref{eqn:p3} holds since we are expanding the time horizon until the end of the window.\\
	Inequality \ref{eqn:p4} follows since $\sum_{n=1}^{i-1} \mathbf{1}\{ \tau_{n}+D_{n}>\tau_{i} \} \geq \sum_{n=1}^{i-1} \mathbf{1}\{ \tau_{n}+D_{n}>t \}$. Specifically, the number of non-reserved ride requests that are received between $\left( kw, \tau_{i}\right]$ and are still active (being served) at time $\tau_{i}$ is \textit{at least as large as} the corresponding number of non-reserved ride requests that are received between $\left( kw, \tau_{i}\right]$ and are still active at time $t\in \left(\tau_{i},  (k+1)w \right]$ (i.e. $t\geq \tau_{i}$).\\
	Then, we can rearrange the last expression in Inequality \ref{eqn:p4} as follows:
	\begin{flalign}
		&P\left( \exists t\in \left(\tau_{i},  (k+1)w \right]: 1+ f_{r}^{P,k}(t) + f_{r}^{BA,k}(t) + \sum_{n=1}^{i-1} \mathbf{1}\{ \tau_{n}+D_{n}>\tau_{i} \} > c^{k}_{r} \right)\label{eqn:p4r}&\\
		&=1-P\left(  1+ f_{r}^{P,k}(t) + f_{r}^{BA,k}(t) + \sum_{n=1}^{i-1} \mathbf{1}\{ \tau_{n}+D_{n}>\tau_{i} \} \leq c^{k}_{r}, \quad \forall t\in \left(\tau_{i},  (k+1)w \right]   \right)\label{eqn:p5}&\\
		&=1-P\left(  1+ \max_{  t\in\left( \tau_{i}, (k+1)w \right] } \left[ f_{r}^{P,k}(t) + f_{r}^{BA,k}(t) \right] + \sum_{n=1}^{i-1} \mathbf{1}\{ \tau_{n}+D_{n}>\tau_{i} \} \leq c^{k}_{r}   \right)\label{eqn:p6}&\\
		&=P\left(  1+ \max_{  t\in\left( \tau_{i}, (k+1)w \right] } \left[ f_{r}^{P,k}(t) + f_{r}^{BA,k}(t) \right] + \sum_{n=1}^{i-1} \mathbf{1}\{ \tau_{n}+D_{n}>\tau_{i} \} > c^{k}_{r}   \right)\label{eqn:p7}&\\
		&=P\left(  \sum_{n=1}^{i-1} \mathbf{1}\{ \tau_{n}+D_{n}>\tau_{i} \} > c^{k}_{r} - \max_{  t\in\left( \tau_{i}, (k+1)w \right] } \left[ f_{r}^{P,k}(t) + f_{r}^{BA,k}(t) \right] - 1  \right)\label{eqn:p8}&\\
		&=P\left(  \sum_{n=1}^{i-1} \mathbf{1}\{ \tau_{n}+D_{n}>\tau_{i} \} \geq c^{k}_{r} - \max_{  t\in\left( \tau_{i}, (k+1)w \right] } \left[ f_{r}^{P,k}(t) + f_{r}^{BA,k}(t) \right]   \right)\label{eqn:p9}&
	\end{flalign}
	Equality \ref{eqn:p6} follows since $f_{r}^{P,k}(t) + f_{r}^{BA,k}(t)$ are the only components that depend on $t$ in expression \ref{eqn:p5}, and if the sum $ 1+ f_{r}^{P,k}(t) + f_{r}^{BA,k}(t) + \sum_{n=1}^{i-1} \mathbf{1}\{ \tau_{n}+D_{n}>\tau_{i} \}$ is less than or equal to $c^{k}_{r}$ at\\ $\displaystyle \tilde{t}=\argmax_{t\in\left( \tau_{i}, (k+1)w \right]}\left[ f_{r}^{P,k}(t) + f_{r}^{BA,k}(t) \right]$, then the aforementioned sum is less than or equal to $c^{k}_{r}$ for all $t\in\left( \tau_{i}, (k+1)w \right]$.\\
	Equality \ref{eqn:p9} follows since $\sum_{n=1}^{i-1} \mathbf{1}\{ \tau_{n}+D_{n}>\tau_{i} \}$, $\max_{  t\in\left( \tau_{i}, (k+1)w \right] } \left[ f_{r}^{P,k}(t) + f_{r}^{BA,k}(t) \right]$, and $c^{k}_{r} $ are all integer values representing the number of active drivers or driver supply. 
	
	Thus,
	\begin{flalign}
		&P(\gamma_{i}=0)\leq P\left(  \sum_{n=1}^{i-1} \mathbf{1}\{ \tau_{n}+D_{n}>\tau_{i} \} \geq c^{k}_{r} - \max_{  t\in\left( \tau_{i}, (k+1)w \right] } \left[ f_{r}^{P,k}(t) + f_{r}^{BA,k}(t) \right]   \right)&
	\end{flalign}
	let $N^{k,\infty}_{r}  \left(  \tau_{i} \right)=\sum_{n=1}^{i-1} \mathbf{1}\{ \tau_{n}+D_{n}>\tau_{i} \}$,\\
	Then,
	\begin{flalign}
		&P(\gamma_{i}=0)\leq P\left(  N^{k,\infty}_{r}  \left(  \tau_{i} \right) \geq c^{k}_{r} - \max_{  t\in\left( \tau_{i}, (k+1)w \right] } \left[ f_{r}^{P,k}(t) + f_{r}^{BA,k}(t) \right]   \right)&
	\end{flalign}
	$N^{k,\infty}_{r}  \left(  \tau_{i} \right)$ represents the number of stochastic non-reserved ride requests that are received between $\left( kw, \tau_{i}\right]$ and are active at time $\tau_{i}$. Thus, $N^{k,\infty}_{r}  \left(  \tau_{i} \right)$ is similar to $N^{k}_{r}\left(  \tau_{i} \right)$ with the main difference being that $N^{k}_{r}\left(  \tau_{i} \right)$ is restricted to admitted non-reserved ride requests while $N^{k,\infty}_{r}  \left(  \tau_{i} \right)$ accounts for \textit{all received requests} (i.e., $N^{k,\infty}_{r}  \left(  \tau_{i} \right)$ assumes that all requests are admitted regardless of the admission control policy). As previously described, stochastic non-reserved ride requests start arriving \textit{after the beginning of the time window} ($t=kw$) according to a Poisson process with demand rate $\{\lambda^{k}_{r}(t):t\in \left( kw, (k+1)w\right] \} $ and their ride duration follows the general distribution $g^{k}_{r}(\cdot)$. Then, the system corresponding to $N^{k,\infty}_{r}  \left(  \tau_{i} \right)$ can be described as a transient  $\text{M}_{t}/\text{GI}/\infty$ queue that \textit{starts empty} at $t=kw$, receives requests at the rate $\{\lambda^{k}_{r}(t):t\in \left( kw, (k+1)w\right] \} $, has a generally distributed service rate $g^{k}_{r}(\cdot)$, and has an infinite number of servers (all requests are admitted). In this context, $N^{k,\infty}_{r}  \left(  \tau_{i} \right)$ (the number of active rides at time $\tau_{i}$) represents the number of busy servers at time $\tau_{i}$ in the transient $\text{M}_{t}/\text{GI}/\infty$ queue.

\end{proof}

\subsection*{B. Minimum Cost Flow Reformulations:}

Original Formulation:
\begin{align}
	&\min_{ h_{ij}: (i,j)\in E,\; h_{i}: i\in R } \qquad \sum_{(i,j)\in E}h_{ij} + M\sum_{i\in R}|h_{i}|\label{eq:obj1-a}\\
	&\textrm{s.t.} \quad \sum_{j:(i,j)\in E}h_{ij} - \sum_{j:(j,i)\in E}h_{ji} + h_{i} = \Delta_{i} \qquad \forall i \in R \label{eq:const11-a}\\
	&\quad  \quad   \sum_{j:(i,j)\in E}h_{ij} \leq e_{i} \qquad \forall i \in R \label{eq:const12-a}\\
	&\quad \quad h_{ij} \geq 0 \qquad \forall (i,j)\in E \label{eq:cons13-a}\\
	&\quad \quad h_{ij} \in \mathbb{Z} \qquad \forall (i,j)\in E \label{eq:cons14-a}\\
	&\quad \quad h_{i} \in \mathbb{Z} \qquad \forall i\in R \label{eq:cons15-a}
\end{align}

First, observe that formulation \ref{eq:obj1-a}--\ref{eq:cons15-a} can be rewritten in terms of $h_{i\bullet}$ and $h_{\bullet i}$ that are defined in Equations \ref{eqn:hbul1-a} and \ref{eqn:hbul2-a}. The revised formulation is given in \ref{eq:obj2-a}--\ref{eq:cons26-a}. In this case, $h_{\bullet i}$ corresponds to drivers added to region $i \in R$ by adjusting the total number of drivers, and $h_{i\bullet}$ corresponds to drivers removed from region $i \in R$ by adjusting the total number of drivers (i.e., $h_{i\bullet}$ represents drivers that can be removed from the system to avoid having excess idle drivers).

\begin{equation}
	\begin{aligned}
		h_{i\bullet} = 
		\begin{cases}
			h_{i} & \text{if} \quad h_{i}>0\\
			0 & \text{otherwise}
		\end{cases}
	\end{aligned}
	\label{eqn:hbul1-a}
\end{equation}

\begin{equation}
	\begin{aligned}
		h_{\bullet i} = 
		\begin{cases}
			|h_{i}| & \text{if} \quad h_{i}<0\\
			0 & \text{otherwise}
		\end{cases}
	\end{aligned}
	\label{eqn:hbul2-a}
\end{equation}

\begin{align}
	&\min_{ h_{ij}: (i,j)\in E,\; h_{i \bullet}, h_{\bullet i}: i\in R } \qquad \sum_{(i,j)\in E}h_{ij} + M\sum_{i\in R} \left[ h_{i \bullet} +  h_{\bullet i}\right] \label{eq:obj2-a}\\
	&\textrm{s.t.} \quad \sum_{j:(i,j)\in E}h_{ij} - \sum_{j:(j,i)\in E}h_{ji} + h_{i \bullet} - h_{\bullet i} = \Delta_{i} \qquad \forall i \in R \label{eq:const21-a}\\
	&\quad  \quad   \sum_{j:(i,j)\in E}h_{ij} \leq e_{i} \qquad \forall i \in R \label{eq:const22-a}\\
	&\quad \quad h_{ij} \geq 0 \qquad \forall (i,j)\in E \label{eq:cons23-a}\\
	&\quad \quad  h_{i \bullet}, h_{\bullet i} \geq 0 \qquad \forall i\in R \label{eq:cons24-a}\\
	&\quad \quad h_{ij} \in \mathbb{Z} \qquad \forall (i,j)\in E \label{eq:cons25-a}\\
	&\quad \quad h_{i \bullet}, h_{\bullet i} \in \mathbb{Z} \qquad \forall i\in R \label{eq:cons26-a}
\end{align}   

Observe that due to the high costs associated with adjusting the total number of drivers, $h_{\bullet i} \leq d^{v}_{i}$ for every region $i$; this inequality implies that the amount of drivers added to region $i$ is less than demand in the region. Similarly, for every region $i$, $h_{i \bullet} \leq s^{v}_{i}$; this inequality implies that the number of drivers disposed from region $i$ (by adjusting the total number of drivers) is less than the virtual supply in the region. If we sum the latter two inequalities over all regions, we get inequalities \ref{eq:sumdem-a} and \ref{eq:sumsup-a}. Then, we can rewrite those inequalities using slack variables as shown in Equations \ref{eq:sumdemslack-a}--\ref{eq:slack-a}.
\begin{align}
	& \sum_{i\in R}h_{\bullet i} \leq \sum_{i\in R}d^{v}_{i} \label{eq:sumdem-a}\\
	& \sum_{i\in R}h_{i \bullet} \leq \sum_{i\in R}s^{v}_{i} \label{eq:sumsup-a}
\end{align} 

\begin{align}
	& \sum_{i\in R}h_{\bullet i} + \bar{h}_{d} = \sum_{i\in R}d^{v}_{i} \label{eq:sumdemslack-a}\\
	& \sum_{i\in R}h_{i \bullet} + \bar{h}_{s} = \sum_{i\in R}s^{v}_{i} \label{eq:sumsupslack-a}\\
	&\bar{h}_{d},\bar{h}_{s} \geq 0 \label{eq:slack-a}
\end{align}     
Intuitively, $\bar{h}_{d}$ is a slack variable that represents the \textit{demand} that is satisfied through internal driver transitions (as opposed to adding external drivers $h_{\bullet i}$ by adjusting the total number of drivers). Meanwhile, $\bar{h}_{s}$ is a slack variable that represents the \textit{supply} that is used to satisfy demand through internal driver transitions (as opposed to disposing off the supply $h_{i \bullet}$ by adjusting the total number of drivers). Therefore, $\bar{h}_{d}=\bar{h}_{s}$. A more rigorous approach to show that the equality holds is as follows:
\begin{lem}
	$\bar{h}_{d}=\bar{h}_{s}=\bar{h}$
\end{lem}
\begin{proof}
	First, we rearrange Equation \ref{eq:sumdemslack-a} to arrive at Equation \ref{eq:cl0-a}. Then, we can restrict the sum to regions where $\Delta_{i}<0$ since by definition $d^{v}_{i}=0$ if $\Delta_{i}\geq 0$, and since $h_{\bullet i} \leq d^{v}_{i}$, then $h_{\bullet i}=0$ if $d^{v}_{i}=0$ (where $h_{\bullet i}\geq 0$ by definition). Thus, $\Delta_{i}\geq 0 \Rightarrow d^{v}_{i}=0 \Rightarrow  h_{\bullet i}=0$, and we can restrict the sum to $\Delta_{i}<0$ as shown in Equation \ref{eq:cl1-a}.\\
	Equation \ref{eq:cl2-a} follows by definition of $d^{v}_{i}$ and $\Delta_{i}$ when $\Delta_{i}<0$.\\
	Equation \ref{eq:cl3-a} follows by rearranging constraint \ref{eq:const21-a}. Note that since $\Delta_{i}<0$ then $s^{v}_{i}=0$ by definition, and since $h_{i \bullet} \leq s^{v}_{i}$ then $h_{i \bullet}=0$.
	\begin{flalign}
		\bar{h}_{d}&=\sum_{i\in R}d^{v}_{i}-h_{\bullet i}\label{eq:cl0-a}&\\ 
		&=\sum_{i\in R:\Delta_{i}<0 }d^{v}_{i}-h_{\bullet i}\label{eq:cl1-a}&\\
		&=\sum_{i\in R:\Delta_{i}<0 }-\Delta_{i}-h_{\bullet i}\label{eq:cl2-a}&\\
		&=\sum_{i\in R:\Delta_{i}<0 }\left[ \sum_{j:(j,i)\in E}h_{ji}-\sum_{j:(i,j)\in E}h_{ij}  \right] \label{eq:cl3-a}&
	\end{flalign}
	Following a similar approach, we can define $\bar{h}_{s}$ as illustrated in Equation \ref{eq:cl4-a}.
	\begin{flalign}
		\bar{h}_{s}&=\sum_{i\in R:\Delta_{i}>0 }\left[\sum_{j:(i,j)\in E}h_{ij} - \sum_{j:(j,i)\in E}h_{ji} \right]\label{eq:cl4-a}&
	\end{flalign}
	Then, we can represent the difference between $\bar{h}_{d}$ and $\bar{h}_{s}$ as in Equation \ref{eq:cl5-a}.\\
	Observe that if $\Delta_{i}=0$, then $\sum_{j:(j,i)\in E}h_{ji} = \sum_{j:(i,j)\in E}h_{ij}$, where this follows by constraint \ref{eq:const21-a} ($h_{i \bullet}=h_{\bullet i}=0$ since $h_{\bullet i} \leq d^{v}_{i}$, $h_{i \bullet} \leq s^{v}_{i}$ and $d^{v}_{i}=s^{v}_{i}=\Delta_{i}=0$).\\
	Thus, we can rearrange Equation \ref{eq:cl5-a} to get Equation \ref{eq:cl6-a}.\\
	Then, we can rearrange Equation \ref{eq:cl6-a} further to get Equations \ref{eq:cl7-a}.
	Finally, note that $\sum_{i\in R }\sum_{j:(j,i)\in E}h_{ji}$ is a summation over all links in the network, and similarly $\sum_{i\in R }\sum_{j:(i,j)\in E}h_{ij} $ is a summation over all links in the network. This gives Equation \ref{eq:cl8-a}, which proves the lemma.
	\begin{flalign}
		\bar{h}_{d} - \bar{h}_{s}&=\sum_{i\in R:\Delta_{i}<0 }\left[ \sum_{j:(j,i)\in E}h_{ji}-\sum_{j:(i,j)\in E}h_{ij}  \right] -  \sum_{i\in R:\Delta_{i}>0 }\left[\sum_{j:(i,j)\in E}h_{ij} - \sum_{j:(j,i)\in E}h_{ji} \right]\label{eq:cl5-a}&\\
		&=\sum_{i\in R }\left[ \sum_{j:(j,i)\in E}h_{ji}-\sum_{j:(i,j)\in E}h_{ij}  \right] \label{eq:cl6-a}\\
		&=\sum_{i\in R }\sum_{j:(j,i)\in E}h_{ji}-\sum_{i\in R }\sum_{j:(i,j)\in E}h_{ij}  \label{eq:cl7-a}\\
		&=\sum_{(i,j)\in E}h_{ij}-\sum_{(i,j)\in E}h_{ij}=0  \label{eq:cl8-a}
	\end{flalign}
\end{proof}
Subsequently, we can add Equations \ref{eq:sumdemslack-a}--\ref{eq:slack-a} as constraints in formulation \ref{eq:obj2-a}--\ref{eq:cons26-a}, where we use $\bar{h}=\bar{h}_{d}=\bar{h}_{s}$. The resulting formulation is shown in \ref{eq:obj3-a}--\ref{eq:cons36ex-a} (Equation \ref{eq:sumsupslack-a} is first multiplied by a negative sign and then added as a constraint). Note that $\bar{h}$ must be integer since, for each region $i$, $s^{v}_{i},d^{v}_{i},h_{\bullet i},h_{i \bullet}$ are all integer.

\begin{align}
	&\min_{ h_{ij}: (i,j)\in E,\; h_{i \bullet}, h_{\bullet i}: i\in R,\;\bar{h} } \qquad \sum_{(i,j)\in E}h_{ij} + M\sum_{i\in R} \left[ h_{i \bullet} +  h_{\bullet i}\right] \label{eq:obj3-a}\\
	&\textrm{s.t.} \quad \sum_{j:(i,j)\in E}h_{ij} - \sum_{j:(j,i)\in E}h_{ji} + h_{i \bullet} - h_{\bullet i} = \Delta_{i} \qquad \forall i \in R \label{eq:const31-a}\\
	&\quad  \quad   \sum_{j:(i,j)\in E}h_{ij} \leq e_{i} \qquad \forall i \in R \label{eq:const32-a}\\
	&\quad  \quad \sum_{i\in R}h_{\bullet i} + \bar{h} = \sum_{i\in R}d^{v}_{i} \label{eq:const32ex1-a}\\
	&\quad  \quad -\left[\sum_{i\in R}h_{i \bullet} + \bar{h}\right] = -\sum_{i\in R}s^{v}_{i} \label{eq:const32ex2-a}\\
	&\quad \quad h_{ij} \geq 0 \qquad \forall (i,j)\in E \label{eq:cons33-a}\\
	&\quad \quad  h_{i \bullet}, h_{\bullet i} \geq 0 \qquad \forall i\in R \label{eq:cons34-a}\\
	&\quad \quad  \bar{h} \geq 0  \label{eq:cons34ex-a}\\
	&\quad \quad h_{ij} \in \mathbb{Z} \qquad \forall (i,j)\in E \label{eq:cons35-a}\\
	&\quad \quad h_{i \bullet}, h_{\bullet i} \in \mathbb{Z} \qquad \forall i\in R \label{eq:cons36-a}\\
	&\quad \quad  \bar{h} \in \mathbb{Z} \label{eq:cons36ex-a}
\end{align}   

To map the problem to an equivalent min-cost flow formulation, for each region $i \in R$, we define variables $h_{ii^{\star}}$ that represent the total number of drivers leaving region $i$ to adjacent regions (Equation \ref{eq:shiistar-a}). In addition, for each link $(i,j)\in E$, we define variables $h_{i^{\star}j}=h_{ij}$. Thus, we can define $h_{ii^{\star}}$ in terms of $h_{i^{\star}j}$ as in Equation  \ref{eq:shiistar2-a}. Since $h_{ij}$ is a non-negative integer for all $(i,j)\in E$, we have that $h_{ii^{\star}}$ and $h_{i^{\star}j}$ are non-negative integers as well.
\begin{align}
	h_{ii^{\star}} &= \sum_{j:(i,j)\in E}h_{ij} \qquad \forall i \in R \label{eq:shiistar-a}\\
	&=\sum_{j:(i,j)\in E}h_{i^{\star}j} \qquad \forall i \in R \label{eq:shiistar2-a}
\end{align} 
Then, we can express constraint \ref{eq:const32-a} in terms of $h_{ii^{\star}}$ as $h_{ii^{\star}}\leq e_{i}$ for all regions $i\in R$. Moreover, we can express the sum of driver transitions across links $(i,j)\in E$ as shown in Equation \ref{eq:newobj-a}. 
\begin{align}
	\sum_{(i,j)\in E}h_{ij}= \sum_{i\in R}\sum_{j:(i,j)\in E}h_{ij}= \sum_{i\in R}h_{ii^{\star}} \label{eq:newobj-a}
\end{align} 
Therefore, we can reformulate optimization problem \ref{eq:obj3-a}--\ref{eq:cons36ex-a} in terms of the newly defined variables as follows: Substitute Equation \ref{eq:newobj-a} in the objective function \ref{eq:obj3-a}, replace the sum of drivers leaving a region to adjacent regions with $h_{ii^{\star}}$ (as in Equation \ref{eq:shiistar-a}), replace $h_{ij}$ by $h_{i^{\star}j}$ and $h_{ji}$ by $h_{j^{\star}i}$, replace constraint \ref{eq:const32-a} with $h_{ii^{\star}}\leq e_{i}$, add Equation \ref{eq:shiistar2-a} to the constraints, add constraints that restrict $h_{i^{\star}j}$ to be non-negative integers for all $(i,j)\in E$, and add constraints that restrict $h_{ii^{\star}}$ to be non-negative integers for all $i \in R$. The revised formulation is shown in \ref{eq:obj4-a}--\ref{eq:cons46ex-a}. 
\begin{align}
	&\min_{ h_{i^{\star}j}: (i,j)\in E,\; h_{i \bullet}, h_{\bullet i}, h_{ii^{\star}}: i\in R,\;\bar{h} } \qquad \sum_{i\in R}h_{ii^{\star}} + M\sum_{i\in R} \left[ h_{i \bullet} +  h_{\bullet i}\right] \label{eq:obj4-a}\\
	&\textrm{s.t.} \quad h_{ii^{\star}} - \sum_{j:(j,i)\in E}h_{j^{\star}i} + h_{i \bullet} - h_{\bullet i} = \Delta_{i} \qquad \forall i \in R \label{eq:const41-a}\\
	&\quad  \quad \sum_{i\in R}h_{\bullet i} + \bar{h} = \sum_{i\in R}d^{v}_{i} \label{eq:const42ex1-a}\\
	&\quad  \quad -\left[\sum_{i\in R}h_{i \bullet} + \bar{h}\right] = -\sum_{i\in R}s^{v}_{i} \label{eq:const42ex2-a}\\
	&\quad  \quad \sum_{j:(i,j)\in E}h_{i^{\star}j} - h_{ii^{\star}} = 0 \qquad \forall i\in R \label{eq:const42ex3-a}\\
	&\quad \quad 0 \leq h_{ii^{\star}} \leq e_{i} \qquad \forall i\in R \label{eq:cons43ext-a}\\
	&\quad \quad h_{i^{\star}j} \geq 0 \qquad \forall (i,j)\in E \label{eq:cons43-a}\\
	&\quad \quad  h_{i \bullet}, h_{\bullet i} \geq 0 \qquad \forall i\in R \label{eq:cons44-a}\\
	&\quad \quad  \bar{h} \geq 0  \label{eq:cons44ex-a}\\
	&\quad \quad h_{ii^{\star}} \in \mathbb{Z} \qquad \forall i\in R \label{eq:cons46extra-a}\\
	&\quad \quad h_{i^{\star}j} \in \mathbb{Z} \qquad \forall (i,j)\in E \label{eq:cons45-a}\\
	&\quad \quad h_{i \bullet}, h_{\bullet i} \in \mathbb{Z} \qquad \forall i\in R \label{eq:cons46-a}\\
	&\quad \quad  \bar{h} \in \mathbb{Z} \label{eq:cons46ex-a}
\end{align}


\newpage

{\footnotesize
\begin{thebibliography}{38}
	\expandafter\ifx\csname natexlab\endcsname\relax\def\natexlab#1{#1}\fi
	\providecommand{\url}[1]{\texttt{#1}}
	\providecommand{\href}[2]{#2}
	\providecommand{\path}[1]{#1}
	\providecommand{\DOIprefix}{doi:}
	\providecommand{\ArXivprefix}{arXiv:}
	\providecommand{\URLprefix}{URL: }
	\providecommand{\Pubmedprefix}{pmid:}
	\providecommand{\doi}[1]{\href{http://dx.doi.org/#1}{\path{#1}}}
	\providecommand{\Pubmed}[1]{\href{pmid:#1}{\path{#1}}}
	\providecommand{\bibinfo}[2]{#2}
	\ifx\xfnm\relax \def\xfnm[#1]{\unskip,\space#1}\fi
	\bibitem[{Ahuja et~al.(1993)Ahuja, Magnanti and Orlin}]{ahuja1993}
	\bibinfo{author}{Ahuja, R.K.}, \bibinfo{author}{Magnanti, T.L.},
	\bibinfo{author}{Orlin, J.B.}, \bibinfo{year}{1993}.
	\newblock \bibinfo{title}{Network {{Flows}}: {{Theory}}, {{Algorithms}}, and
		{{Applications}}}.
	\newblock \bibinfo{publisher}{{Prentice Hall}}.
	\bibitem[{Bahat and Bekhor(2016)}]{bahat2016}
	\bibinfo{author}{Bahat, O.}, \bibinfo{author}{Bekhor, S.},
	\bibinfo{year}{2016}.
	\newblock \bibinfo{title}{Incorporating ridesharing in the static traffic
		assignment model}.
	\newblock \bibinfo{journal}{Networks and Spatial Economics}
	\bibinfo{volume}{16}, \bibinfo{pages}{1125--1149}.
	\bibitem[{Ban et~al.(2019)Ban, Dessouky, Pang and Fan}]{ban2019}
	\bibinfo{author}{Ban, X.}, \bibinfo{author}{Dessouky, M.},
	\bibinfo{author}{Pang, J.}, \bibinfo{author}{Fan, R.}, \bibinfo{year}{2019}.
	\newblock \bibinfo{title}{A general equilibrium model for transportation
		systems with e-hailing services and flow congestion}.
	\newblock \bibinfo{journal}{Transportation Research Part B: Methodological}
	\bibinfo{volume}{129}, \bibinfo{pages}{273--304}.
	\bibitem[{Banerjee et~al.(2017)Banerjee, Freund and Lykouris}]{banerjee2017}
	\bibinfo{author}{Banerjee, S.}, \bibinfo{author}{Freund, D.},
	\bibinfo{author}{Lykouris, T.}, \bibinfo{year}{2017}.
	\newblock \bibinfo{title}{Pricing and optimization in shared vehicle systems:
		{{An}} approximation framework}.
	\newblock \bibinfo{journal}{arXiv preprint} .
	\bibitem[{Banerjee et~al.(2018)Banerjee, Kanoria and Qian}]{banerjee2018}
	\bibinfo{author}{Banerjee, S.}, \bibinfo{author}{Kanoria, Y.},
	\bibinfo{author}{Qian, P.}, \bibinfo{year}{2018}.
	\newblock \bibinfo{title}{State dependent control of closed queueing networks
		with application to ride-hailing}.
	\newblock \bibinfo{journal}{arXiv preprint} .
	\bibitem[{Braverman et~al.(2019)Braverman, Dai, Liu and Ying}]{braverman2019}
	\bibinfo{author}{Braverman, A.}, \bibinfo{author}{Dai, J.G.},
	\bibinfo{author}{Liu, X.}, \bibinfo{author}{Ying, L.}, \bibinfo{year}{2019}.
	\newblock \bibinfo{title}{Empty-car routing in ridesharing systems}.
	\newblock \bibinfo{journal}{Operations Research} \bibinfo{volume}{67},
	\bibinfo{pages}{1437--1452}.
	\bibitem[{Chen et~al.(2019)Chen, Zhang, Liu and Nie}]{chen2019}
	\bibinfo{author}{Chen, H.}, \bibinfo{author}{Zhang, K.}, \bibinfo{author}{Liu,
		X.}, \bibinfo{author}{Nie, Y.M.}, \bibinfo{year}{2019}.
	\newblock \bibinfo{title}{A physical model of street ride-hail}.
	\newblock \bibinfo{journal}{SSRN 3318557} .
	\bibitem[{Daganzo and Ouyang(2019)}]{daganzo2019}
	\bibinfo{author}{Daganzo, C.F.}, \bibinfo{author}{Ouyang, Y.},
	\bibinfo{year}{2019}.
	\newblock \bibinfo{title}{A general model of demand-responsive transportation
		services: {{From}} taxi to ridesharing to dial-a-ride}.
	\newblock \bibinfo{journal}{Transportation Research Part B: Methodological}
	\bibinfo{volume}{126}, \bibinfo{pages}{213--224}.
	\bibitem[{Di and Ban(2019)}]{di2019}
	\bibinfo{author}{Di, X.}, \bibinfo{author}{Ban, X.J.}, \bibinfo{year}{2019}.
	\newblock \bibinfo{title}{A unified equilibrium framework of new shared
		mobility systems}.
	\newblock \bibinfo{journal}{Transportation Research Part B: Methodological}
	\bibinfo{volume}{129}, \bibinfo{pages}{50--78}.
	\bibitem[{Di et~al.(2018)Di, Ma, Liu and Ban}]{di2018}
	\bibinfo{author}{Di, X.}, \bibinfo{author}{Ma, R.}, \bibinfo{author}{Liu, H.},
	\bibinfo{author}{Ban, X.J.}, \bibinfo{year}{2018}.
	\newblock \bibinfo{title}{A link-node reformulation of ridesharing user
		equilibrium with network design}.
	\newblock \bibinfo{journal}{Transportation Research Part B: Methodological}
	\bibinfo{volume}{112}, \bibinfo{pages}{230--255}.
	\bibitem[{Djavadian and Chow(2017)}]{djavadian2017}
	\bibinfo{author}{Djavadian, S.}, \bibinfo{author}{Chow, J.Y.},
	\bibinfo{year}{2017}.
	\newblock \bibinfo{title}{An agent-based day-to-day adjustment process for
		modeling `{{Mobility}} as a {{Service}}' with a two-sided flexible transport
		market}.
	\newblock \bibinfo{journal}{Transportation research part B: methodological}
	\bibinfo{volume}{104}, \bibinfo{pages}{36--57}.
	\bibitem[{Eick et~al.(1993)Eick, Massey and Whitt}]{eick1993}
	\bibinfo{author}{Eick, S.G.}, \bibinfo{author}{Massey, W.A.},
	\bibinfo{author}{Whitt, W.}, \bibinfo{year}{1993}.
	\newblock \bibinfo{title}{The physics of the $\text{M}_{t}$/{{G}}/$\infty$
		queue}.
	\newblock \bibinfo{journal}{Operations Research} \bibinfo{volume}{41},
	\bibinfo{pages}{731--742}.
	\bibitem[{Foley(1982)}]{foley1982}
	\bibinfo{author}{Foley, R.D.}, \bibinfo{year}{1982}.
	\newblock \bibinfo{title}{The nonhomogeneous {{M}}/{{G}}/$\infty$ queue}.
	\newblock \bibinfo{journal}{Opsearch} \bibinfo{volume}{19},
	\bibinfo{pages}{40--48}.
	\bibitem[{Lei et~al.(2019)Lei, Jiang and Ouyang}]{lei2019}
	\bibinfo{author}{Lei, C.}, \bibinfo{author}{Jiang, Z.},
	\bibinfo{author}{Ouyang, Y.}, \bibinfo{year}{2019}.
	\newblock \bibinfo{title}{Path-based dynamic pricing for vehicle allocation in
		ridesharing systems with fully compliant drivers}.
	\newblock \bibinfo{journal}{Transportation Research Part B: Methodological}
	\bibinfo{volume}{(forthcoming)}.
	\bibitem[{Li et~al.(2019)Li, Tavafoghi, Poolla and Varaiya}]{li2019}
	\bibinfo{author}{Li, S.}, \bibinfo{author}{Tavafoghi, H.},
	\bibinfo{author}{Poolla, K.}, \bibinfo{author}{Varaiya, P.},
	\bibinfo{year}{2019}.
	\newblock \bibinfo{title}{Regulating {{TNCs}}: {{Should Uber}} and {{Lyft}} set
		their own rules?}
	\newblock \bibinfo{journal}{Transportation Research Part B: Methodological}
	\bibinfo{volume}{129}, \bibinfo{pages}{193--225}.
	\bibitem[{Lyft(2019a)}]{Lyft2019a}
	\bibinfo{author}{Lyft}, \bibinfo{year}{2019}a.
	\newblock \bibinfo{title}{Bonuses and {{Incentives}}}.
	\newblock
	\bibinfo{howpublished}{https://help.lyft.com/hc/en-us/sections/115003494568-Bonuses-and-Incentives}.
	\bibitem[{Lyft(2019b)}]{Lyft2019}
	\bibinfo{author}{Lyft}, \bibinfo{year}{2019}b.
	\newblock \bibinfo{title}{New {{York City Driver Information}}}.
	\newblock
	\bibinfo{howpublished}{https://help.lyft.com/hc/en-us/articles/115012929447-New-York-City-Driver-Information}.
	\bibitem[{Lyft(2019c)}]{Lyft2019b}
	\bibinfo{author}{Lyft}, \bibinfo{year}{2019}c.
	\newblock \bibinfo{title}{Prime {{Time}} for drivers}.
	\newblock
	\bibinfo{howpublished}{https://help.lyft.com/hc/en-us/articles/115012926467-Prime-Time-for-drivers}.
	\bibitem[{Nie(2017)}]{nie2017}
	\bibinfo{author}{Nie, Y.M.}, \bibinfo{year}{2017}.
	\newblock \bibinfo{title}{How can the taxi industry survive the tide of
		ridesourcing? {{Evidence}} from {{Shenzhen}}, {{China}}}.
	\newblock \bibinfo{journal}{Transportation Research Part C: Emerging
		Technologies} \bibinfo{volume}{79}, \bibinfo{pages}{242--256}.
	\bibitem[{Nourinejad and Ramezani(2019)}]{nourinejad2019}
	\bibinfo{author}{Nourinejad, M.}, \bibinfo{author}{Ramezani, M.},
	\bibinfo{year}{2019}.
	\newblock \bibinfo{title}{Ride-{{Sourcing}} modeling and pricing in
		non-equilibrium two-sided markets}.
	\newblock \bibinfo{journal}{Transportation Research Part B: Methodological}
	\bibinfo{volume}{(forthcoming)}.
	\bibitem[{NYCTLC(2019)}]{nyctlc2019}
	\bibinfo{author}{NYCTLC}, \bibinfo{year}{2019}.
	\newblock \bibinfo{title}{{{TLC}} {{Trip Record Data}}}.
	\newblock
	\bibinfo{howpublished}{https://www1.nyc.gov/site/tlc/about/tlc-trip-record-data.page}.
	\bibitem[{Ozkan and Ward(2019)}]{ozkan2019}
	\bibinfo{author}{Ozkan, E.}, \bibinfo{author}{Ward, A.}, \bibinfo{year}{2019}.
	\newblock \bibinfo{title}{Dynamic matching for real-time ridesharing}.
	\newblock \bibinfo{journal}{Stochastic Systems}
	\bibinfo{volume}{(forthcoming)}.
	\bibitem[{Pr{\'e}kopa(1958)}]{prekopa1958}
	\bibinfo{author}{Pr{\'e}kopa, A.}, \bibinfo{year}{1958}.
	\newblock \bibinfo{title}{On secondary processes generated by a random point
		distribution of {{Poisson}} type}.
	\newblock \bibinfo{journal}{Annales Univ. Sci. Budapest de E{\"o}tv{\"o}s Nom.
		Sectio Math} \bibinfo{volume}{1}, \bibinfo{pages}{153--170}.
	\bibitem[{Qian and Ukkusuri(2017)}]{qian2017a}
	\bibinfo{author}{Qian, X.}, \bibinfo{author}{Ukkusuri, S.V.},
	\bibinfo{year}{2017}.
	\newblock \bibinfo{title}{Taxi market equilibrium with third-party hailing
		service}.
	\newblock \bibinfo{journal}{Transportation Research Part B: Methodological}
	\bibinfo{volume}{100}, \bibinfo{pages}{43--63}.
	\bibitem[{Ramezani and Nourinejad(2018)}]{ramezani2018}
	\bibinfo{author}{Ramezani, M.}, \bibinfo{author}{Nourinejad, M.},
	\bibinfo{year}{2018}.
	\newblock \bibinfo{title}{Dynamic modeling and control of taxi services in
		large-scale urban networks: {{A}} macroscopic approach}.
	\newblock \bibinfo{journal}{Transportation Research Part C: Emerging
		Technologies} \bibinfo{volume}{94}, \bibinfo{pages}{203--219}.
	\bibitem[{Rasulkhani and Chow(2019)}]{rasulkhani2019}
	\bibinfo{author}{Rasulkhani, S.}, \bibinfo{author}{Chow, J.Y.},
	\bibinfo{year}{2019}.
	\newblock \bibinfo{title}{Route-cost-assignment with joint user and operator
		behavior as a many-to-one stable matching assignment game}.
	\newblock \bibinfo{journal}{Transportation Research Part B: Methodological}
	\bibinfo{volume}{124}, \bibinfo{pages}{60--81}.
	\bibitem[{Wang and Yang(2019)}]{wang2019b}
	\bibinfo{author}{Wang, H.}, \bibinfo{author}{Yang, H.}, \bibinfo{year}{2019}.
	\newblock \bibinfo{title}{Ridesourcing systems: {{A}} framework and review}.
	\newblock \bibinfo{journal}{Transportation Research Part B: Methodological}
	\bibinfo{volume}{129}, \bibinfo{pages}{122--155}.
	\bibitem[{Wang et~al.(2019)Wang, Ban and Huang}]{wang2019}
	\bibinfo{author}{Wang, J.P.}, \bibinfo{author}{Ban, X.J.},
	\bibinfo{author}{Huang, H.J.}, \bibinfo{year}{2019}.
	\newblock \bibinfo{title}{Dynamic ridesharing with variable-ratio
		charging-compensation scheme for morning commute}.
	\newblock \bibinfo{journal}{Transportation Research Part B: Methodological}
	\bibinfo{volume}{122}, \bibinfo{pages}{390--415}.
	\bibitem[{Wang et~al.(2018)Wang, Yang and Zhu}]{wang2018}
	\bibinfo{author}{Wang, X.}, \bibinfo{author}{Yang, H.}, \bibinfo{author}{Zhu,
		D.}, \bibinfo{year}{2018}.
	\newblock \bibinfo{title}{Driver-rider cost-sharing strategies and equilibria
		in a ridesharing program}.
	\newblock \bibinfo{journal}{Transportation Science} \bibinfo{volume}{52},
	\bibinfo{pages}{868--881}.
	\bibitem[{Wolsey(1998)}]{wolsey1998}
	\bibinfo{author}{Wolsey, L.}, \bibinfo{year}{1998}.
	\newblock \bibinfo{title}{Integer {{Programming}}}.
	\newblock \bibinfo{publisher}{{Wiley}}.
	\bibitem[{Xu et~al.(2019)Xu, Yin and Ye}]{xu2019}
	\bibinfo{author}{Xu, Z.}, \bibinfo{author}{Yin, Y.}, \bibinfo{author}{Ye, J.},
	\bibinfo{year}{2019}.
	\newblock \bibinfo{title}{On the supply curve of ride-hailing systems}.
	\newblock \bibinfo{journal}{Transportation Research Part B: Methodological}
	\bibinfo{volume}{(forthcoming)}.
	\bibitem[{Yang and Yang(2011)}]{yang2011}
	\bibinfo{author}{Yang, H.}, \bibinfo{author}{Yang, T.}, \bibinfo{year}{2011}.
	\newblock \bibinfo{title}{Equilibrium properties of taxi markets with search
		frictions}.
	\newblock \bibinfo{journal}{Transportation Research Part B: Methodological}
	\bibinfo{volume}{45}, \bibinfo{pages}{696--713}.
	\bibitem[{Zha et~al.(2018)Zha, Yin and Du}]{zha2018a}
	\bibinfo{author}{Zha, L.}, \bibinfo{author}{Yin, Y.}, \bibinfo{author}{Du, Y.},
	\bibinfo{year}{2018}.
	\newblock \bibinfo{title}{Surge pricing and labor supply in the ride-sourcing
		market}.
	\newblock \bibinfo{journal}{Transportation Research Part B: Methodological}
	\bibinfo{volume}{117}, \bibinfo{pages}{708--722}.
	\bibitem[{Zha et~al.(2016)Zha, Yin and Yang}]{zha2016}
	\bibinfo{author}{Zha, L.}, \bibinfo{author}{Yin, Y.}, \bibinfo{author}{Yang,
		H.}, \bibinfo{year}{2016}.
	\newblock \bibinfo{title}{Economic analysis of ride-sourcing markets}.
	\newblock \bibinfo{journal}{Transportation Research Part C: Emerging
		Technologies} \bibinfo{volume}{71}, \bibinfo{pages}{249--266}.
	\bibitem[{Zhang et~al.(2019)Zhang, Chen, Yao, Xu, Ge, Liu and Nie}]{zhang2019b}
	\bibinfo{author}{Zhang, K.}, \bibinfo{author}{Chen, H.}, \bibinfo{author}{Yao,
		S.}, \bibinfo{author}{Xu, L.}, \bibinfo{author}{Ge, J.},
	\bibinfo{author}{Liu, X.}, \bibinfo{author}{Nie, M.}, \bibinfo{year}{2019}.
	\newblock \bibinfo{title}{An efficiency paradox of uberization}.
	\newblock \bibinfo{journal}{SSRN 3462912} .
	\bibitem[{Zhang and Nie(2019)}]{zhang2019a}
	\bibinfo{author}{Zhang, K.}, \bibinfo{author}{Nie, M.}, \bibinfo{year}{2019}.
	\newblock \bibinfo{title}{To pool or not to pool: {{Equilibrium}}, pricing and
		regulation}.
	\newblock \bibinfo{journal}{SSRN 3497808} .
	\bibitem[{Zhang and Pavone(2016)}]{zhang2016}
	\bibinfo{author}{Zhang, R.}, \bibinfo{author}{Pavone, M.},
	\bibinfo{year}{2016}.
	\newblock \bibinfo{title}{Control of robotic mobility-on-demand systems: A
		queueing-theoretical perspective}.
	\newblock \bibinfo{journal}{The International Journal of Robotics Research}
	\bibinfo{volume}{35}, \bibinfo{pages}{186--203}.
	\bibitem[{{Zuniga-Garcia} et~al.(2020){Zuniga-Garcia}, Tec, Scott, {Ruiz-Juri}
		and Machemehl}]{zuniga-garcia2020}
	\bibinfo{author}{{Zuniga-Garcia}, N.}, \bibinfo{author}{Tec, M.},
	\bibinfo{author}{Scott, J.G.}, \bibinfo{author}{{Ruiz-Juri}, N.},
	\bibinfo{author}{Machemehl, R.B.}, \bibinfo{year}{2020}.
	\newblock \bibinfo{title}{Evaluation of ride-sourcing search frictions and
		driver productivity: {{A}} spatial denoising approach}.
	\newblock \bibinfo{journal}{Transportation Research Part C: Emerging
		Technologies} \bibinfo{volume}{110}, \bibinfo{pages}{346--367}.
	
\end{thebibliography}
}

\end{document}